\documentclass[online]{tlp}
\usepackage[arxiv]{preamble}
\usepackage{fixkeywords}

\ifarxiv
\usepackage{hyperref}
\fi

\usepackage{amsfonts}
\usepackage{amssymb}
\usepackage{amsmath}
\usepackage{latexsym}
\usepackage{stmaryrd}
\usepackage{microtype}
\usepackage{booktabs}
\usepackage{bigdelim}

\usepackage{horn} 
\usepackage{customcommands}  

\begin{document}

\lefttitle{A. Charalambidis, B. Kostopoulos, Ch. Nomikos and P. Rondogiannis}

\jnlPage{1}{8}
\jnlDoiYr{2021}
\doival{10.1017/xxxxx}

\title[The Power of Negation in Higher-Order Datalog]{The Power of Negation in Higher-Order Datalog
  \thanks{This work was supported by a research project which is implemented in the framework of H.F.R.I call ``Basic research
          Financing (Horizontal support of all Sciences)'' under the National Recovery and Resilience
          Plan ``Greece 2.0'' funded by the European Union - NextGenerationEU (H.F.R.I. Project Number: 16116).}}

\begin{authgrp}
\author{%
      \gn{Angelos} \sn{Charalambidis},
      \gn{Babis} \sn{Kostopoulos}}
\affiliation{Harokopio University of Athens, Greece}
\author{%
      \gn{Christos} \sn{Nomikos}}
\affiliation{University of Ioannina, Greece}
\author{%
      \gn{Panos} \sn{Rondogiannis}}
\affiliation{National and Kapodistrian University of Athens, Greece}
\end{authgrp}

\history{\sub{xx xx xxxx;} \rev{xx xx xxxx;} \acc{xx xx xxxx}}

\maketitle
\begin{abstract}
We investigate the expressive power of Higher-Order Datalog$^\neg$
under both the well-founded and the stable model semantics, establishing tight
connections with complexity classes. We prove that under the well-founded
semantics, for all $k\geq 1$, $(k+1)$-Order Datalog$^\neg$ captures \EXPTIME[k], a result that
holds without explicit ordering of the input database. The proof of this
fact can be performed either by using the powerful existential predicate
variables of the language or by using partially applied relations
and relation enumeration. Furthermore, we demonstrate that this
expressive power is retained within a stratified fragment of the language. Under
the stable model semantics, we show that $(k+1)$-Order Datalog$^\neg$ captures
\coNEXPTIME[k] using cautious reasoning and \NEXPTIME[k] using brave reasoning,
again with analogous results for the stratified fragment augmented with choice
rules. Our results establish a hierarchy of expressive power, highlighting an
interesting trade-off between order and non-determinism in the context of
higher-order logic programming: increasing the order of programs under the
well-founded semantics can surpass the expressive power of
lower-order programs under the stable model semantics.
\end{abstract}


\section{Introduction}
The superior expressive power of higher-order functional programming languages
with respect to their first-order counterparts has been thoroughly demonstrated
by~\cite{jones}. In logic programming, the first result of this type was
established by~\cite{iclp19}, where it was demonstrated that positive
Higher-Order Datalog programs can capture broader complexity classes as their
order increases. In particular, it was demonstrated that
for all $k\geq 1$, $(k+1)$-Order Datalog captures $\EXPTIME[k]$, under the assumption
that the input database is \emph{ordered}. The
aforementioned result generalized a classical expressibility theorem which
states that (first-order) Datalog captures
$\PTIME$~\citep{Pap85,Gra92,Var82,Imm86,Lei89}, again under the assumption that
the input database is ordered. Notice that the ordering assumption underlying the above results,
is actually a rather strong one because it allows (even weak) declarative query languages
to simulate the ordering of the tape of a Turing machine.

Recently,~\cite{iclp24} defined the well-founded and the stable model semantics of
Higher-Order Datalog with negation and illustrated its expressive power with
non-trivial examples. Remarkably, one such example was the Generalized Geography
two-player game, which is a well-known~\citep{LS80GOPolynomial-SpaceHard}
$\mathsf{PSPACE}$-complete problem. The examples given in \cite{iclp24} do not
seem to require any ordering of the input and solve the corresponding problems
declaratively.
As noted by~\cite{iclp24}, such examples
indicate ``\emph{that higher-order logic programming under the stable model semantics
is a powerful and versatile formalism, which can potentially form the basis of novel ASP
systems}''.
Such systems may be able to cope with demanding problems that arise in combinatorial optimization, game theory, machine learning
theory, and so on (see, for example, the discussion in~\cite{ART19BeyondNPQuantifyingoverAnswerSets,BJT16Stable-unstablesemanticsBeyondNPnormallogic}).
The results of~\cite{iclp24} trigger the natural question of the exact
characterization of the expressive power of Higher-Order Datalog with negation
and of whether we can obtain expressibility results that do not rely on the ordering
assumption.

In this paper we undertake the formal study of the expressive power of
Higher-Order Datalog with negation -- which in the rest of the paper is denoted
by Higher-Order Datalog$^\neg$. The results we obtain indeed demonstrate that
negation is a very powerful construct of the language, justifying the increased
expressiveness which was conjectured by~\cite{iclp24}. Our results, which will
be explained in detail in the rest of the paper, are presented in
Table~\ref{table1}. All the results in the table are new, with the exception of
those results concerning programs of order 1 (which are well-known, see for
example~\cite{dantsin,Niemela08}). In particular, the main contributions of the
paper can be summarized as follows:

\begin{table}
\centering
\caption{Expressive Power of Higher-Order Datalog$^\neg$ (with no ordering assumption).}
\label{table1}
{\tablefont\begin{tabular}{@{\extracolsep{\fill}}p{4.9cm}ccccccl@{}}
\topline
{\bf Fragment} & \multicolumn{1}{c}{} & {\bf Semantics}
                                          &  \multicolumn{4}{c}{\bf Order of the program} \\
\cmidrule(lr){4-7}
 & & & 1 & 2 & $\cdots$ & $k+1$
\midline
HO-Datalog$^{\neg}$ & \rdelim\}{3}{*} &  &  &  &  &  \\
HO-Datalog$^{\neg,\not\exists}$ & & Well-Founded &$\subsetneq$ \textsf{P} & \textsf{EXP} & $\cdots$ & $k$-\textsf{EXP} \\
Stratified HO-Datalog$^{\neg,\not\exists}$ & &  &  &  &  &
\midline
HO-Datalog$^{\neg}$ & \rdelim\}{3}{*} &  &  &  & &  \\
HO-Datalog$^{\neg,\not\exists}$ & & Stable (cautious) & \textsf{co}-\textsf{NP} & \textsf{co}-\textsf{NEXP} & $\cdots$ & \textsf{co}-($k$-\textsf{NEXP}) \\
Stratified+Choices HO-Datalog$^{\neg,\not\exists}$ & &  &  &  &  &
\midline
HO-Datalog$^{\neg}$ & \rdelim\}{3}{*} &  &  &  &  &  \\
HO-Datalog$^{\neg,\not\exists}$ & & Stable (brave) & \textsf{NP} & \textsf{NEXP} & $\cdots$ & $k$-\textsf{NEXP} \\
Stratified+Choices HO-Datalog$^{\neg,\not\exists}$ & &  & &  &  &
\botline
\end{tabular}}
\end{table}

\begin{itemize}
\item We establish that $(k+1)$-Order Datalog$^\neg$ captures \EXPTIME[k] under
      the well-founded semantics. Notably, this result holds without requiring a
      predefined ordering of the input database: the use of existential
      predicate variables, facilitates the
      construction of such an ordering. Furthermore, we show that even in the
      fragment of the language without existential predicate variables (denoted
      by HO-Datalog$^{\neg,\not\exists}$ in Table~\ref{table1}), the same result
      can be achieved through partial applications and an enumeration procedure.
      Perhaps even more strikingly, the \EXPTIME[k] expressibility result also
      holds for a stratified fragment of $(k+1)$-Order Datalog$^{\neg,\not\exists}$.
      This last result is especially unexpected, considering that the stratified
      fragment of classical (first-order) Datalog$^\neg$ exhibits strictly lower
      expressive power than Datalog$^\neg$ under the well-founded
      semantics~\citep{Kolaitis91}.

\item We demonstrate that $(k+1)$-Order Datalog$^\neg$ captures \coNEXPTIME[k]
      under the stable model semantics using cautious reasoning and \NEXPTIME[k]
      under brave reasoning. As before, these two results hold with or without the
      presence of existential predicate variables. Additionally,
      the two results hold within a fragment of $(k+1)$-Order Datalog$^\neg$
      that consists of programs that have a stratified part together with a
      simple unstratified part of a very specific form (\emph{choice rules}). In
      other words, we prove that the expressive power, under the stable model
      semantics, of $(k+1)$-Order Datalog$^\neg$ is equivalent to
      the power of the aforementioned restricted fragment.

\item Since it is well-known that $\EXPTIME[(k-1)] \subseteq \NEXPTIME[(k-1)] \subseteq \EXPTIME[k]$ and
      $\EXPTIME[(k-1)] \subseteq \coNEXPTIME[(k-1)] \subseteq \EXPTIME[k]$,
      $k$-Order Datalog$^\neg$ programs under the well-founded semantics are
      at most as powerful as $k$-Order Datalog$^\neg$ programs under the stable
      model semantics, which, in turn, are at most as powerful as $(k+1)$-Order
      Datalog$^\neg$ programs under the well-founded semantics (under both the brave and cautious reasoning schemes). This observation
      illustrates an interesting trade-off between order and non-determinism in
      the context of higher-order logic programming: by increasing the order of
      our programs while using well-founded semantics, we can
      surpass the expressive power provided by non-determinism in lower-order
      programs under the stable model semantics.
\end{itemize}
The rest of the paper is structured as follows: Section~\ref{preliminaries}
introduces the language we will be studying.
Section~\ref{wellfounded-simulation} derives the expressive power of
Higher-Order Datalog$^\neg$ under the well-founded semantics and
Section~\ref{stable-simulation} the power under the stable model semantics.
Section~\ref{transformation} presents a semantics-preserving transformation that
eliminates existential predicate variables from clause bodies; this
transformation implies that our results hold even without the presence of
existential predicate variables. Finally, Section~\ref{conclusions} concludes
the paper giving pointers for future work.

\section{Higher-Order Datalog with Negation: Preliminaries}\label{preliminaries}


In this section we define the syntax of Higher-Order Datalog$^\neg$.
The language uses two base types: $\bool$, the Boolean domain, and $\basedom$,
the domain of data objects. The composite types are partitioned into
\emph{predicate} ones (assigned to predicate symbols) and \emph{argument} ones
(assigned to parameters of predicates).
\begin{definition}
Types are either \emph{predicate} or \emph{argument}, denoted by $\pi$
and $\rho$ respectively, and defined as:
\begin{align*}
\pi  & := \bool \mid (\rho \to \pi) \\
\rho & := \basedom \mid \pi
\end{align*}
\end{definition}

As usual, the binary operator $\to$ is right-associative. It can be easily seen
that every predicate type $\pi$ can be written in the form
$\rho_1 \to \cdots \to \rho_n \rightarrow \bool$,
$n\geq 0$ (for $n=0$ we assume that $\pi=\bool$).
%
%

\begin{definition}
The \emph{alphabet} of Higher-Order Datalog$^\neg$ consists of: \emph{predicate variables}
of every predicate type $\pi$ (denoted by capital letters such as $\mathtt{P},\mathtt{Q},\ldots)$;
\emph{predicate constants} of every predicate type $\pi$ (denoted by lowercase letters such as $\mathtt{p},\mathtt{q},\ldots$);
\emph{individual variables} of type $\basedom$ (denoted by capital letters such as $\mathtt{X},\mathtt{Y},\ldots$);
\emph{individual constants} of type $\basedom$ (denoted by lowercase letters such as $\mathtt{a},\mathtt{b},\ldots$);
the \emph{equality}  constant $\approx$ of type $\basedom \to \basedom \to o$;
the \emph{conjunction} constant $\wedge$ of type $\bool \to \bool \to \bool$;
the \emph{inverse implication} constant $\lrule$ of type $\bool \to \bool \to \bool$; and
the \emph{negation} constant ${\tt not}$ of type $\bool \to \bool$.
\end{definition}

Arbitrary variables (either predicate or individual ones) will be denoted by $\mathsf{R}$. 

\begin{definition}
\label{def:expressions}

The \emph{expressions} and \emph{literals} of Higher-Order Datalog$^\neg$ are
defined as follows.
Every predicate variable/constant and every individual variable/constant is an
expression of the corresponding type; if $\mathsf{E}_1$ is an expression of type
$\rho \to \pi$ and $\mathsf{E}_2$ an expression of type $\rho$ then
$(\mathsf{E}_1\ \mathsf{E}_2)$ is an expression of type $\pi$.
Every expression of type $\bool$ is called an \emph{atom}.
If $\mathsf{E}$ is an atom, then $\mathsf{E}$ and $({\tt not}\, \mathsf{E})$
are literals of type $\bool$; if $\mathsf{E}_1$ and
$\mathsf{E}_2$ are expressions of type $\basedom$, then $(\mathsf{E}_1\approx
\mathsf{E}_2)$ and ${\tt not}\, (\mathsf{E}_1\approx \mathsf{E}_2)$
are literals of type $\bool$.
\end{definition}

We will omit parentheses when no confusion arises.

\begin{definition}
A \emph{rule} of Higher-Order Datalog$^\neg$ is a formula
$\mathsf{p}\ \mathsf{R}_1 \cdots \mathsf{R}_n \lrule \mathsf{L}_1 \land \ldots \land \mathsf{L}_m$,
where $\mathsf{p}$ is a predicate constant of type $\rho_1 \to \cdots \to \rho_n \to \bool$,
$\mathsf{R}_1,\ldots,\mathsf{R}_n$ are distinct variables of types $\rho_1,\ldots,\rho_n$ respectively and
the $\mathsf{L}_i$ are literals.
The literal $\mathsf{p}\ \mathsf{R}_1 \cdots \mathsf{R}_n$ is the \emph{head} of the rule and
$ \mathsf{L}_1 \land \ldots \land \mathsf{L}_m$ is the \emph{body} of the rule.
A \emph{program} $\Prog$ of Higher-Order Datalog$^\neg$ is a finite set of rules.
\end{definition}

We will follow the common logic programming practice and write
$\mathsf{L}_1,\ldots,\mathsf{L}_m$ instead of
$\mathsf{L}_1 \wedge \cdots \wedge \mathsf{L}_m$ for the body of a rule.
For brevity reasons, we will often denote
a rule as $\mathsf{p} \ \overline{\mathsf{R}} \lrule \mathsf{B}$, where
$\overline{\mathsf{R}}$ is a shorthand for a sequence of variables
$\mathsf{R}_1 \cdots \mathsf{R}_n$ and $\mathsf{B}$ represents the body of the rule.
By abuse of notation, in the programs that we will write, we will avoid using currying as much as possible and
will use tuples instead, a syntax that is more familiar to logic programmers.
The tuple syntax can be directly transformed to the curried one by a simple
preprocessing. So, for example, instead of {\tt succ Ord X Y} we will write {\tt succ(Ord,X,Y)},
instead of the partial application {\tt succ Ord} we will write {\tt succ(Ord)},
and so on. More generally, the partial application
$\mathsf{p}\ \mathsf{E}_1\ \cdots\ \mathsf{E}_n$ will be written as
$\mathsf{p}(\mathsf{E}_1,\ldots,\mathsf{E}_n)$.

The well-founded and the stable model semantics of Higher-Order Datalog$^\neg$,
were defined in~\cite{iclp24}. The main idea of the semantics is to interpret
higher-order user-defined predicate constants as three-valued relations over two-valued
objects, \ie as functions that take classical relations as arguments and return $\mtrue$,
$\mfalse$, or $\mundef$. This interpretation of predicate constants, apart from giving a simple
denotation of the various constructs of the language, also allows one to use Approximation Fixpoint
Theory (AFT)~\citep{DMT00ApproximationsStableOperatorsWell-FoundedFixpointsApplications,DMT04Ultimateapproximationapplicationnonmonotonicknowledgerepresentation}, in order to define a variety of semantics for the language (such as
well-founded, stable, Kripke-Kleene, and so on). For the reader of the main part of the present paper,
a deep understanding of this semantics is not necessary: the programs that we give can be understood
purely declaratively, without resorting to the help of the semantics (in the same way that a
logic programmer does not need to master its model-theoretic semantics in order to write or understand a program).
Programs of our language simply define extensional higher-order relations.
For example, in the well-founded semantics a unary second-order predicate simply denotes a relation that takes a classical set as an argument and returns $\mtrue$,
$\mfalse$, or $\mundef$ as the result; actually, most of our programs will be \emph{stratified}
(see the forthcoming Definition~\ref{stratified}), which means that they actually denote classical
higher-order relations (\ie they never return $\mundef$ as the result).
Moreover, in the stable model semantics all the predicates denote two-valued relations.
\ifincludeappendix
The only points where the
reader will have to delve deeper into the semantics, is in order to understand the proofs of certain theorems
given in the appendices.
For this reason (and also due to space restrictions)
the full presentation of the semantics is given in Appendix~\ref{semantics}.
\else
The only points where the
reader will have to delve deeper into the semantics, is in order to understand the proofs of certain theorems
provided in an appendix as supplementary material.
For this reason (and also due to space restrictions)
the full presentation of the semantics is also given as supplementary material.
\fi


The notion of \emph{order} of a predicate, is formally defined as follows:
\begin{definition}
The \emph{order} of a type is recursively defined as follows:
\[
\begin{array}{rcl}
\textit{order}(\iota) & = & 0 \\
\textit{order}(o)     & = & 1\\
\textit{order}(\rho \to \pi) & = & \max\{ \mathit{order}(\rho) + 1,
                                          \mathit{order}(\mathit{\pi}) \}\\
\end{array}
\]
The order of a predicate constant (or variable) is the order of its type.
\end{definition}
\begin{definition}
For all $k\geq 1$, \emph{$k$-Order Datalog$^\neg$} is the fragment of Higher-Order Datalog$^\neg$
in which all variables have order less than or equal to $k-1$ and all predicate constants
in the program have order less than or equal to $k$.
\end{definition}

The following example showcases many aspects of the language and additionally introduces some useful predicates that will be needed in our subsequent simulations.
\begin{example}\label{example1}
We define the relation
\lstinline`hamilton(X,Y)` which is true iff there exists a Hamilton path from vertex
\lstinline`X` to vertex \lstinline`Y` in a graph represented by a binary predicate
\lstinline`e` which specifies the edges of the graph.
The first rule in the definition of \lstinline`hamilton` is the following:
\begin{lstlisting}
hamilton(X,Y):-ordering(Ord),first(Ord,X),last(Ord,Y),subset(succ(Ord),e).
\end{lstlisting}
The above rule states that there exists a Hamilton path from vertex
\lstinline`X` to vertex \lstinline`Y` if there exists a relation \lstinline`Ord`
that is a strict total ordering with first element \lstinline`X` and last element
\lstinline`Y` and for every two consecutive elements in \lstinline`Ord` the
corresponding edge exists in \lstinline`e`. Notice the use of \lstinline`Ord` in
the above rule: it is a predicate variable that does not appear in the head of
the rule and therefore it is an existentially quantified variable of the body
(\ie the body can be read as ``there exists a relation \lstinline`Ord` such that
\ldots''). To be a strict total ordering, \lstinline`Ord` must be irreflexive,
transitive, and every two different elements must be related. This can be
expressed with the following rules:
\begin{lstlisting}
ordering(Ord):-connected(Ord),transitive(Ord),irreflexive(Ord).
connected(Ord):-not disconnected(Ord).
disconnected(Ord):-not Ord(X,Y),not Ord(Y,X),not(X=Y).
transitive(Ord):-not non_transitive(Ord).
non_transitive(Ord):-Ord(X,Y),Ord(Y,Z),not Ord(X,Z).
irreflexive(Ord):-not non_irreflexive(Ord).
non_irreflexive(Ord):-Ord(X,X).
\end{lstlisting}
It is interesting to note above how we can implement
universal quantification in the body of a rule: for example, to express the fact
that \lstinline`Ord` is connected, \ie that \emph{for all} \lstinline`X`,
\lstinline`Y` either \lstinline`X` is related to \lstinline`Y` or vice-versa, we
just require that \lstinline`Ord` is not disconnected, \ie it is not the case
that there exist \lstinline`X`, \lstinline`Y` that are not related. This is a
common trick that we will use throughout the paper in order to represent
universal quantification.

We now define the predicates \lstinline`first`, \lstinline`last` and \lstinline`succ`.
Predicate \lstinline`first(Ord,X)` is true for \lstinline`X` being the individual constant that
is the first element with respect to the ordering specified by \lstinline`Ord`.
Likewise, \lstinline`last(Ord,X)` is true if \lstinline`X` is the last element
in \lstinline`Ord`. The predicate \lstinline`succ(Ord,X,Y)` is true for \lstinline`X`
and \lstinline`Y` that are sequential in \lstinline`Ord`.
\begin{lstlisting}
first(Ord,X):-not nfirst(Ord,X).
nfirst(Ord,X):-Ord(Z,X).
last(Ord,X):-not nlast(Ord,X).
nlast(Ord,X):-Ord(X,Y).
succ(Ord,X,Y):-Ord(X,Y),not nsequential(Ord,X,Y).
nsequential(Ord,X,Y):-Ord(X,Z),Ord(Z,Y).
\end{lstlisting}
Finally, we have the rules for \lstinline`subset`:
\begin{lstlisting}
subset(P,Q):-not nonsubset(P,Q).
nonsubset(P,Q):-P(X),not Q(X).
nonsubset(P,Q):-not P(X), Q(X).
\end{lstlisting}
The above two rules use again the trick for implementing universal quantification.%
\hfill\hbox{\proofbox}
\end{example}

As in the case of first-order logic programs with negation, there exists a
simple notion of stratification for higher-order logic programs with
negation~\citep{iclp24}.
\begin{definition}\label{stratified}
A program $\Prog$ is called \emph{stratified} if there exists a function $S$ mapping
predicate constants to natural numbers, such that for each rule
$\mathsf{p} \ \overline{\mathsf{R}} \leftarrow \mathsf{L}_1,\ldots,\mathsf{L}_m$
and any $i\in\{1,\ldots, m\}$:
\begin{itemize}
  \item $S(\mathsf{q})\leq S(\mathsf{p})$ for every predicate constant
        $\mathsf{q}$ occurring in $\mathsf{L}_i$.
  \item If $\mathsf{L}_i$ is of the form $({\tt not}\,\,\mathsf{E})$, then
        $S(\mathsf{q})< S(\mathsf{p})$ for each predicate constant $\mathsf{q}$
        occurring~in~$\mathsf{E}$.
  \item For any subexpression of $\mathsf{L}_i$ of the form
        $(\mathsf{E}_1~\mathsf{E}_2)$, $S(\mathsf{q})< S(\mathsf{p})$ for every
        predicate constant $\mathsf{q}$ occurring in $\mathsf{E}_2$.
\end{itemize}
\end{definition}

A possibly unexpected aspect of the above definition, is the last item, which
says that the stratification function should not only increase because of
negation, but also because of higher-order predicate application. The intuitive
reason for this is that in Higher-Order Datalog$^\neg$ one can define a
higher-order predicate which is identical to negation, for example, by writing
\lstinline|neg P :- not P|. As a consequence, it is reasonable to assume that
predicates occurring inside an application of \lstinline|neg| should be treated
similarly to predicates appearing inside the negation symbol.

One can easily verify that the Hamiltonian Path program of
Example~\ref{example1}, is stratified.
As we are going to see, the expressive power of stratified programs is the same
as that of the non-stratified ones for orders $k \geq 2$ under the well-founded semantics.

Languages such as Higher-Order Datalog$^\neg$, are usually referred as
\emph{formal query languages}. A program in our language can be considered to
compute a query in the following sense: a first-order predicate, like
\lstinline`e` in Example~\ref{example1}, will be called an
\emph{input predicate} and its denotation (as a set of ground atoms) constitutes what is
called the \emph{input database}, usually denoted by $D_{in}$; a first-order
predicate like \lstinline`hamilton` in Example~\ref{example1}, will be an
\emph{output} one and its denotation constitutes the \emph{output database},
usually denoted by $D_{out}$.

More formally, a \emph{database schema} $\sigma$ is a finite set of first-order
predicate symbols with associated arities. A \emph{database} over a schema $\sigma$
is a finite set of ground atoms whose predicate symbols belong to $\sigma$.
A \emph{query} is a mapping from databases over a schema $\sigma_1$ to databases over a schema $\sigma_2$.
%
A program $\mathsf{P}$ can be seen as a query ${\cal Q}_{\mathsf{P}}$ such that
$D_{out} = {\cal Q}_{\mathsf{P}}(D_{in})$. We are interested in queries that are
\emph{generic}~\citep{Imm86}, \ie queries that do not depend on the names of the
individual constants in the input database.
Given a fragment of our language, we are interested in the \emph{expressive
power} of the fragment under a given semantics, namely the set of queries that can be defined by
programs of the fragment. In particular, we want to demonstrate that such a
fragment \emph{captures a complexity class $\mathcal{C}$}, \ie it can express
\emph{exactly} all the queries whose \emph{evaluation complexity} belongs to
$\mathcal{C}$. By evaluation complexity we mean the complexity of checking
whether a given atom belongs to the output database. Notice that different
semantics of the fragments we study may lead to different evaluation
complexities and therefore to capturing different complexity classes. Therefore,
our results will be of the form ``\emph{the fragment X of Higher-Order
Datalog$^\neg$, under the Y semantics, captures the complexity class Z}''.

In this work we consider the well-founded semantics, the stable model semantics with
cautions reasoning and the stable model semantics with brave reasoning. In particular,
under the well-founded semantics, a given atom belongs to the output database if and only if
it is true in the well-founded model. Moreover, under the stable model semantics with cautious (resp. brave)
reasoning, a given atom belongs to the output database if and only if it is true
in all stable models (resp. in at least one stable model).





\section{Expressive Power under the Well-founded Semantics}\label{wellfounded-simulation}
The purpose of this section is to demonstrate that for all $k\geq 1$,
$(k+1)$-Order Datalog$^\neg$, under the well-founded semantics, captures
\EXPTIME[k]. Proofs of such results usually consist of two parts; in our case
these two parts can be intuitively described as follows:
\begin{itemize}
\item We show that every Turing machine that takes as input an encoding of an
      input database and computes a query over this relation that belongs to
      \EXPTIME[k], can be simulated by a $(k+1)$-Order
      Datalog$^\neg$ program (whose meaning is understood under the well-founded
      semantics).

\item We show that computing the well-founded semantics of every $(k+1)$-Order
      Datalog$^\neg$ program over an input database, can be done in
      $k$-exponential time with respect to the number $n$ of individual constants in
      the input database.
\end{itemize}

In the main part of the paper, we focus on the proof of the first result.
\ifincludeappendix
The proof of the second result is given in Appendix~\ref{appendixB}.
\else
The proof of the second result is given in an appendix as supplementary material.
\fi
The proof of the first
result has two important points that have to receive special attention, on which
we briefly comment below.

\vspace{0.1cm}

\noindent
{\bf Ordering of the input database:}
A Turing machine encodes an input database as a string on its tape. Such an
encoding provides an implicit strict total ordering of the input. On the other hand,
Higher-Order Datalog$^\neg$ does not have a notion of tape, and therefore the
input database is, at first sight, unordered. In less powerful languages, such
as for example the language of~\cite{iclp19}, this mismatch is solved by imposing an
\emph{explicit} strict total ordering on the individual constants of the input
database.
However, it turns out that in Higher-Order Datalog$^\neg$
we do not need this ordering trick (usually referred in the
literature as the \emph{ordering assumption}): we can generate an ordering
\lstinline`Ord` of the constants of the input database and then use it to
perform the Turing machine simulation.

\vspace{0.1cm}

\noindent{\bf Representing numbers:}
Since we want to simulate the operation of a Turing machine, we need to have a
numbering scheme in our language in order to count the steps of the Turing
machine and the positions on its tape. At first, we need some base-numbers;
actually, it is sufficient to use the $n$ individual constants in the input
database. The trick here is to use these constants directly as numbers: we
generate an ordering \lstinline`Ord` on these constants and require that it is a
strict total order. We then show how to simulate bigger numbers, namely numbers
polynomially related to $n$. Finally, since we simulate the operation of a
Turing machine that runs in $k$-exponential time, we must be able to represent
bigger numbers. The trick here is to use higher-order relations: as the order of
our programs increases, the more numbers we can represent. Actually, as we
demonstrate, $(k+1)$-Order Datalog$^\neg$ is sufficient
to represent $k$-exponential numbers.

\vspace{0.1cm}

Regarding the ordering of the input database, we already have all the required
machinery ready from Example~\ref{example1}. We can use the predicate
\lstinline`ordering` to generate a \emph{strict total order} \lstinline`Ord`,
over the constant symbols of the input database, which we will then use in all
our simulations. In other words, the constants of our input database play, under
this ordering, the role of \emph{base-numbers} in our simulation (\ie ``small''
numbers that can get up to $n-1$).
Given an individual constant $\mathsf{c}$ of the input database, we will say it
represents the natural number $m$, or more formally $num(\mathsf{c})=m$, if and
only if in the strict total order \lstinline`Ord` the constant $\mathsf{c}$ is
the $(m+1)$-th element. In the following, we see how we can use these
base-numbers to represent even bigger ones.


\subsection{Representing Numbers}

\paragraph{\bf Representing Polynomially-Big Numbers:}
To represent natural numbers up to $n^{d+1}-1$, where $d$ is any arbitrary but fixed natural number,
we use tuples of individual constants with fixed length of size $d+1$.
The following predicates define the
``first'' and ``last'' of such numbers (denoted by \lstinline`first$_0$` and \lstinline`last$_0$`)
and the ``less-than'' and ``successor'' relations on them (denoted by \lstinline`lt$_0$`
and \lstinline`succ$_0$`). Notice that the following definitions use the
\lstinline`first` and \lstinline`last` predicates defined in Example~\ref{example1}.

\begin{lstlisting}
first$_0$(Ord,X$_0$,...,X$_d$):-first(Ord,X$_0$),...,first(Ord,X$_d$).
last$_0$(Ord,X$_0$,...,X$_d$):-last(Ord,X$_0$),...,last(Ord,X$_d$).
lt$_0$(Ord,X$_0$,...,X$_d$,Y$_0$,...,Y$_d$):-Ord(X$_d$,Y$_d$).
lt$_0$(Ord,X$_0$,...,X$_d$,Y$_0$,...,Y$_d$):-Ord(X$_{d-1}$,Y$_{d-1}$),(X$_d$=Y$_d$).
$\ldots$
lt$_0$(Ord,X$_0$,...,X$_d$,Y$_0$,...,Y$_d$):-Ord(X$_0$,Y$_0$),(X$_1$=Y$_1$),...,(X$_d$=Y$_d$).
succ$_0$(Ord,$\bar{\tt X}$,$\bar{\tt Y}$):-lt$_0$(Ord,$\bar{\tt X}$,$\bar{\tt Y}$),not nsequential$_0$(Ord,$\bar{\tt X}$,$\bar{\tt Y}$).
nsequential$_0$(Ord,$\bar{\tt X}$,$\bar{\tt Y}$):-lt$_0$(Ord,$\bar{\tt X}$,$\bar{\tt Z}$),lt$_0$(Ord,$\bar{\tt Z}$,$\bar{\tt Y}$).
\end{lstlisting}

In the tuple-representation of numbers, tuple $({\tt X}_0,\ldots,{\tt X}_d)$ of
base-elements represents the number $num{({\tt X}_0,\ldots,{\tt X}_d)} =
num({\tt X}_0) + num({\tt X}_1)\cdot n + \cdots + num({\tt X}_d)\cdot n^{d}$.


% ...








\paragraph{\bf Representing Exponentially-Big Numbers:}
We now demonstrate how we can represent ``exponentially-big'' numbers as
higher-order relations. In the foregoing discussion
we will need the following notation:
$\expk{0}(x) = x$ and $\expk{n+1}(x) = 2^{\expk{n}(x)}$.
Let $N_0 = n^{d+1} - 1$ be the largest number that can be represented
by $(d+1)$-tuples of individual constants and for $k \geq 1$,
let $N_k$ be the largest number that can be represented by using
$k$-order relations.
We can exponentially increase the numbers up to the
number $N_{k+1} = \exp_1(N_k + 1)-1$ by using $(k+1)$-order relations.
One can easily see that $N_{k} = \exp_{k}(n^{d+1})-1$.

If the $k$-order relations representing numbers up to $N_k$ are of type $\rho$,
then it suffices to use higher-order relations of type $\rho \rightarrow o$ in order
to represent numbers up to $N_{k+1}$. This is essentially a binary representation
where the lower order numbers denote bit positions. Formally, let  \lstinline`Z` be a
$(k+1)$-order element and ${\tt R}_0,\ldots,{\tt R}_{N_k}$ be the ordering of
the elements that represent numbers in the previous counting module. Let $f$ be the
function mapping $\mtrue$ to $1$ and $\mfalse$ to $0$. Then we have
$num({\tt Z}) = f({\tt Z(R}_{0})) + f({\tt Z(R}_{1})) \cdot 2 + \cdots + f({\tt Z(R}_{N_k}))\cdot 2^{N_k}$.
We begin with predicates testing for the first and the last number.
\begin{lstlisting}
first$_{k+1}$(Ord,N):-not nfirst$_{k+1}$(Ord,N).
nfirst$_{k+1}$(Ord,N):-N(X).
last$_{k+1}$(Ord,N):-not nlast$_{k+1}$(Ord,N).
nlast$_{k+1}$(Ord,N):-not N(X).
\end{lstlisting}

The following definitions describe the ``less than'' relation between two
elements that represent numbers. We examine if a number \lstinline`N` is less
than \lstinline`M` by comparing the two numbers bit by bit in their binary
representation. The successor of a number is defined with the use of less-than.
%








\begin{lstlisting}
lt$_{k+1}$(Ord,N,M):-last$_{k}$(Ord,X),bit$_{k+1}$(Ord,N,M,X).
bit$_{k+1}$(Ord,N,M,X):-not N(X),M(X).
bit$_{k+1}$(Ord,N,M,X):-N(X),M(X),succ$_{k}$(Ord,Y,X),bit$_{k+1}$(Ord,N,M,Y).
bit$_{k+1}$(Ord,N,M,X):-not N(X),not M(X),succ$_{k}$(Ord,Y,X),bit$_{k+1}$(Ord,N,M,Y).
succ$_{k+1}$(Ord,N,M):-lt$_{k+1}$(Ord,N,M), not nsequential$_{k+1}$(Ord,N,M).
nsequential$_{k+1}$(Ord,N,M):-lt$_{k+1}$(Ord,N,Z),lt$_{k+1}$(Ord,Z,M).
\end{lstlisting}
For $k=0$ the above code is slightly different: variables \lstinline`X` and \lstinline`Y`
should be replaced with \lstinline`$\bar{\tt X}$` and \lstinline`$\bar{\tt Y}$`.

\subsection{Turing machine simulation}
%


We now demonstrate how any query that belongs to \EXPTIME[k] ($k\geq 1$), can be
expressed by a $(k+1)$-Order Datalog$^\neg$ program under the well-founded
semantics. It suffices to assume that the output schema of the query consists of
a single output predicate, since every query can be decomposed into multiple
queries of this form. Since the query belongs to \EXPTIME[k], there exists a
Turing machine that given on its tape an input database under some sensible encoding,
decides whether a tuple belongs to the output relation of
the query in at most $\exp_{k}(n^d)$ steps, where $n$ is the number of constant
symbols in the input database and $d$ is some constant. We simulate this Turing
machine with a $(k+1)$-Order Datalog$^\neg$ program.

\paragraph{\bf Encoding the input:}
Before presenting the simulation of the Turing machine $M$, we mention certain
simplifying assumptions, which do not affect the generality of the subsequent
results.
\begin{itemize}
  \item The input database consists of a single binary relation \lstinline|in|
        and the output database is also a single binary relation
        \lstinline|out|. In the following, the number of constants in the input
        database is denoted by $n$.
  \item The alphabet of $M$ that will be simulated is $\Sigma =\{0,1,\Box\}$.
        $M$ expects the input relation \lstinline|in| as the standard binary
        encoding of a graph, which is based on the ordering of the individual
        constants, in the first $n^2$ cells of its tape. For example, if the
        pair $(x,y)$ belongs to \lstinline|in|, then the tape of $M$ contains a
        ``1'' at cell position $num(x)+num(y) \cdot n$, otherwise it contains
        ``0''.
  \item $M$ decides whether a tuple $(a,b)$ belongs to the output relation
        \lstinline|out|. The next $n^2$ cells of its tape are used to encode
        $(a,b)$. All these cells contain the symbol ``0'', except for the cell
        at position $num(a)+num(b) \cdot n+n^2$ which contains ``1''.
  \item $M$ reaches its accepting state $yes$ if and only if the tuple $(a,b)$ belongs to the output relation
        \lstinline|out|.
\end{itemize}

The following two predicates encode the binary input relation \lstinline|in| and the tuple $(a,b)$ as a binary string.
\begin{lstlisting}
input$_{1}$(A,B,Ord,X,Y,Z$_2$,...,Z$_d$):-first(Ord,Z$_2$),...,first(Ord,Z$_d$),in(X,Y).
input$_{0}$(A,B,Ord,X,Y,Z$_2$,...,Z$_d$):-first(Ord,Z$_2$),...,first(Ord,Z$_d$),not in(X,Y).
input$_{1}$(A,B,Ord,Z$_0$,Z$_1$,Z$_2$,...,Z$_d$):-(Z$_0$=A),(Z$_1$=B),first(Ord,Z),succ(Ord,Z,Z$_2$),
                               first(Ord,Z$_3$),...,first(Ord,Z$_d$).
input$_{0}$(A,B,Ord,Z$_0$,Z$_1$,Z$_2$,...,Z$_d$):-not(Z$_0$=A),first(Ord,Z),succ(Ord,Z,Z$_2$),
                               first(Ord,Z$_3$),...,first(Ord,Z$_d$).
input$_{0}$(A,B,Ord,Z$_0$,Z$_1$,Z$_2$,...,Z$_d$):-not(Z$_1$=B),first(Ord,Z),succ(Ord,Z,Z$_2$),
                               first(Ord,Z$_3$),...,first(Ord,Z$_d$).
\end{lstlisting}
The predicate \lstinline|input$_{1}$(A,B,Ord,Z$_0$,...,Z$_d$)| is true if the
symbol \lstinline|1| will be written during the initialization of $M$ in the
cell of the tape represented by the number \lstinline|(Z$_0$,...,Z$_d$)|.
Similarly, \lstinline|input$_{0}$(A,B,Ord,Z$_0$,...,Z$_d$)| is true if the
symbol \lstinline|0| will be written in that position.

\paragraph{\bf Initial configuration of the Turing Machine:}
In order to represent the configurations of the Turing machine we use a
higher-order predicate for each state and symbol, and a higher-order predicate
for the cursor position. The predicate \lstinline|state$_{s}$(A,B,Ord,T)| is
true if at time \lstinline|T| the Turing machine is in state $s$. The predicate
\lstinline|symbol$_{\sigma}$(A,B,Ord,T,P)|, where  $\sigma \in \{0,1,\Box\}$, is
true if at time \lstinline|T| the tape has symbol $\sigma$ written in position
\lstinline|P|. The predicate \lstinline|cursor(A,B,Ord,T,P)| is true if at time
\lstinline|T| the cursor is in position \lstinline|P|.

We also need a higher-order predicate to lift the tuple representation of numbers to
the $k$-order numbering notation. Predicate
\lstinline|lift$_{k}$(Ord,$\bar{\tt X}$,M)| transforms the number represented by
the tuple \lstinline|$\bar{\tt X}$|
in the zero-order notation to the same number \lstinline`M` in the $k$-order notation.
\begin{lstlisting}
lift$_{k}$(Ord,$\bar{\tt X}$,M):-first$_0$(Ord,$\bar{\tt X}$),first$_{k}$(Ord,M).
lift$_{k}$(Ord,$\bar{\tt X}$,M):-succ$_0$(Ord,$\bar{\tt Z}$,$\bar{\tt X}$),succ$_{k}$(Ord,M',M),lift$_{k}$(Ord,$\bar{\tt Z}$,M').
\end{lstlisting}

We proceed to describe the initial configuration of the Turing machine. At time
0, the machine is in its initial state $s_0$, the tape contains only the binary
string of the encoded input and all the other positions are filled with the
symbol $\Box$, and the cursor is in position 0.
This is described by the following rules:
\begin{lstlisting}
state$_{s_0}$(A,B,Ord,T):-first$_{k}$(Ord,T).
cursor(A,B,Ord,T,P):-first$_{k}$(Ord,T),first$_{k}$(Ord,P).
symbol$_{0}$(A,B,Ord,T,P):-first$_{k}$(Ord,T),input$_{0}$(A,B,Ord,$\bar{\tt X}$),lift$_{k}$(Ord,$\bar{\tt X}$,P).
symbol$_{1}$(A,B,Ord,T,P):-first$_{k}$(Ord,T),input$_{1}$(A,B,Ord,$\bar{\tt X}$),lift$_{k}$(Ord,$\bar{\tt X}$,P).
symbol$_{\Box}$(A,B,Ord,T,P):-first$_{k}$(Ord,T),not symbol$_{0}$(A,B,Ord,T,P),
                      not symbol$_{1}$(A,B,Ord,T,P).
\end{lstlisting}
We now specify how the execution of the Turing machine is simulated.

\paragraph{\bf Simulating transitions:}
To describe the transitions of the machine, we create rules which assert the
next configuration of the machine, based on the current one and the transition
function.
Let $(s,\sigma)\rightarrow(s',\sigma',\text{right})$ be a transition, indicating
that if the current state of the machine is $s$ and its cursor reads the symbol
$\sigma$, then the machine changes its state to $s'$, it writes $\sigma'$ on the
current cursor position and then moves the cursor to the right.
In our simulation we use an auxiliary predicate
\lstinline|current$_{(s,\sigma)}$(A,B,Ord,T,P)| which is true if at the moment
\lstinline`T` the state of the machine is $s$, the cursor position is
\lstinline`P` and the cursor reads the symbol $\sigma$.
\begin{lstlisting}
current$_{(s,\sigma)}$(A,B,Ord,T,P):-state$_{s}$(A,B,Ord,T),cursor(A,B,Ord,T,P),
                         symbol$_{\sigma}$(A,B,Ord,T,P).
\end{lstlisting}
Now, the transition can be simulated by the following rules:
\begin{lstlisting}
state$_{s'}$(A,B,Ord,T'):-current$_{(s,\sigma)}$(A,B,Ord,T,P),succ$_{k}$(Ord,T,T').
symbol$_{\sigma'}$(A,B,Ord,T',P):-current$_{(s,\sigma)}$(A,B,Ord,T,P),succ$_{k}$(Ord,T,T').
cursor(A,B,Ord,T',P'):-current$_{(s,\sigma)}$(A,B,Ord,T,P),
                     succ$_{k}$(Ord,P,P'),succ$_{k}$(Ord,T,T').
\end{lstlisting}
Other types of transitions can be expressed in a similar way.
We also add what is commonly called ``inertia rules''. They ensure that
every position of the tape except for the position of the cursor, retains its content.
For any symbol $\sigma$ we include the following:
\begin{lstlisting}
symbol$_{\sigma}$(A,B,Ord,T',P'):-succ$_{k}$(Ord,T,T'),symbol$_{\sigma}$(A,B,Ord,T,P'),
                      cursor(A,B,Ord,T,P),lt$_{k}$(Ord,P,P').
symbol$_{\sigma}$(A,B,Ord,T',P'):-succ$_{k}$(Ord,T,T'),symbol$_{\sigma}$(A,B,Ord,T,P'),
                      cursor(A,B,Ord,T,P),lt$_{k}$(Ord,P',P).
\end{lstlisting}

Finally, in order to produce the output relation, we use the following:
\begin{lstlisting}
out(A,B):-ordering(Ord),state$_{yes}$(A,B,Ord,T).
\end{lstlisting}
The \lstinline`out` predicate is the one that ``initiates the simulation'': it
produces, using an existential predicate variable, the ordering \lstinline`Ord`
that is used throughout the simulation and it verifies that there exists a value \lstinline`T`
that represents a time point (within the range of representable numbers) such that the machine reaches an accepting state. The above
simulation leads to the following theorem:
\begin{theorem}
Every query in \EXPTIME[k] ($k \geq 1$) can be expressed by a $(k+1)$-Order
Datalog$^\neg$ program.
%
\end{theorem}

The opposite direction of the above theorem also holds, as stated by the
following theorem\ifincludeappendix (see Appendix~\ref{appendixB} for the proof)\fi.
\begin{theorem}\label{turing-machine-well-founded}
%
Let $\Prog$ be a $(k+1)$-Order Datalog$^\neg$ program that defines a query
$\mathcal{Q}_\mathsf{P}$ under the well-founded semantics.
Then, there exists a deterministic Turing
machine that takes as input an encoding of a database $D$
that uses $n$ individual constant symbols and
a ground atom $p(\bar{a})$, where $p$ is a predicate constant of $\Prog$ and
$\bar{a}$ is a tuple of those individual constants, and
decides whether $p(\bar{a}) \in \mathcal{Q}_\mathsf{P}(D)$, in at most
$exp_{k}(n^d)$ steps for some constant $d$.
\end{theorem}

By inspecting the simulation program, one easily sees that it is stratified. Therefore,
we get the following result:
\begin{corollary}
$(k+1)$-Order Datalog$^\neg$ and Stratified $(k+1)$-Order Datalog$^\neg$ capture
\EXPTIME[k] under the well-founded semantics.
\end{corollary}

%

As a consequence, Higher-Order Datalog$^\neg$ and Stratified Higher-Order
Datalog$^\neg$ both capture \ELEMENTARY{} (\ie the union of \EXPTIME[k] for all $k$).

\section{Expressive Power under the Stable model Semantics}\label{stable-simulation}
In this section, we study the expressiveness of Higher-Order Datalog$^\neg$ under the
stable model semantics.
%
We demonstrate how any query that belongs to \coNEXPTIME[k] ($k\geq 1$), can be
expressed by a $(k+1)$-Order Datalog$^\neg$ program under the stable model
semantics with cautious reasoning. Since the query belongs to \coNEXPTIME[k],
there exists a non-deterministic Turing machine $M$ that decides whether a tuple
belongs to the complement of the output relation of the query in at most
$\exp_{k}(n^d)$ steps, where $n$ is the number of constants in the input
database and $d$ is a constant. Without loss of generality, we assume that
each computational path of $M$ terminates after at most $\exp_{k}(n^d)$ steps at
a state in \{yes, no\}. We simulate $M$ with a
$(k+1)$-Order Datalog$^\neg$ program.

For the most part, the simulation is the same as that of
Section~\ref{wellfounded-simulation}. It is intuitively helpful to consider each
stable model of the following simulation as a possible computation path the
machine could have taken. We will add some additional ``choice'' rules to
simulate those non-deterministic transitions at any possible time step.

Let $(\sigma,s)$ be a pair of a symbol and a state of the Turing machine and
assume there exist $m$ possible transitions from this pair. Since the machine is
non-deterministic, $m$ can be greater than $1$. We add the following predicates
and rules to our program.

\begin{lstlisting}
b$_{\sigma,s,m}$(T):-not b$_{\sigma,s,1}$(T),...,not b$_{\sigma,s,m-1}$(T).
b$_{\sigma,s,m-1}$(T):-not b$_{\sigma,s,1}$(T),...,not b$_{\sigma,s,m-2}$(T),not b$_{\sigma,s,m}$(T).
$\ldots$
b$_{\sigma,s,1}$(T):-not b$_{\sigma,s,2}$(T),...,not b$_{\sigma,s,m}$(T).
\end{lstlisting}

This ensures that in every stable model and for every time point \lstinline|T|,
exactly one of \lstinline|b$_{\sigma,s,i}$(T)|, $i \in \{ 1,\ldots,m \}$, is true.

Let also the transition table be
$(\sigma,s) \rightarrow (\sigma'_i,s'_i,\text{move}_i)$ for $i = 1,\ldots,m$.
Like in the deterministic case we create rules for each one such transition but
we also add the previous branching predicates so that only exactly one rule can
be ``active'' in a stable model.

\begin{lstlisting}
state$_{s'_i}$(A,B,Ord,T'):-b$_{\sigma,s,i}$(T),current$_{(s,\sigma)}$(A,B,Ord,T,P),succ$_{k}$(Ord,T,T').
symbol$_{\sigma'_i}$(A,B,Ord,T',P):-b$_{\sigma,s,i}$(T),current$_{(s,\sigma)}$(A,B,Ord,T,P),succ$_{k}$(Ord,T,T').
\end{lstlisting}
Assuming, for example, that $\text{move}_i$ is ``right'' (and likewise for
``left'' and ``stay''):
\begin{lstlisting}
cursor(A,B,T',P'):-b$_{\sigma,s,i}$(T),current$_{(s,\sigma)}$(A,B,Ord,T,P),
                   succ$_{k}$(Ord,T,T'),succ$_{k}$(Ord,P,P').
\end{lstlisting}

If every computational path of the Turing machine reaches state ``no'', then
for every stable model there exists a time point \lstinline|T| such that \lstinline|state$_{no}$(A,B,Ord,T)| is true.
Therefore, the following rule defines the output relation under cautious reasoning:
\begin{lstlisting}
out(A,B):-ordering(Ord),state$_{no}$(A,B,Ord,T).
\end{lstlisting}

The above discussion leads to the following theorem:
\begin{theorem}
Every query in \coNEXPTIME[k] can be expressed by a $(k+1)$-Order Datalog$^\neg$
program under the stable model semantics and cautious reasoning.
\end{theorem}
\ifincludeappendix
The following theorem is the converse of the previous one and its proof can be
found in Appendix~\ref{appendixB}.
\else
The following theorem is the converse of the previous one and its proof can be
found in the supplementary material.
\fi
\begin{theorem}\label{Turing-machine-stable-model}
Let $\Prog$ be a $(k+1)$-Order Datalog$^\neg$ program that defines a query
$\mathcal{Q}_\mathsf{P}$ under the stable model semantics and cautious reasoning.
Then, there exists a non-deterministic Turing
machine that takes as input an encoding of a database $D$
that uses $n$ individual constant symbols and
a ground atom $p(\bar{a})$, where $p$ is a predicate constant of $\Prog$ and
$\bar{a}$ is a tuple of those individual constants, and
decides whether $p(\bar{a}) \not\in \mathcal{Q}_\mathsf{P}(D)$, in at most
$exp_{k}(n^d)$ steps for some constant $d$.
\end{theorem}

Notice that our simulation consists of a stratified program together with the
rules that define the \lstinline`b$_{\sigma,s,m}$`. We call this fragment of
Higher-Order Datalog$^\neg$ ``\emph{Stratified+Choices Higher-Order
Datalog$^\neg$}''. We therefore have the following result.
\begin{corollary}
$(k+1)$-Order Datalog$^\neg$ and Stratified+Choices $(k+1)$-Order Datalog$^\neg$
capture \coNEXPTIME[k] under the stable model semantics and cautious reasoning.
\end{corollary}

By a similar kind of analysis we can derive a result for the stable model
semantics under brave reasoning. The arguments are very similar and omitted.
\begin{corollary}
$(k+1)$-Order Datalog$^\neg$ and Stratified+Choices $(k+1)$-Order Datalog$^\neg$
capture \NEXPTIME[k] under the stable model semantics and brave reasoning.
\end{corollary}

Since $\EXPTIME[(k-1)] \subseteq \NEXPTIME[(k-1)] \subseteq \EXPTIME[k]$ and
$\EXPTIME[(k-1)] \subseteq \coNEXPTIME[(k-1)] \subseteq \EXPTIME[k]$ it follows
that Higher-order Datalog$^\neg$ has the same expressive power
(namely \ELEMENTARY) under both the well-founded and the stable model semantics, in
both reasoning schemes.

\section{Removing Higher-Order Existential Predicate Variables\label{transformation}}

The Turing machine simulations discussed in the previous sections rely
extensively on existential predicate variables. Actually, the simulations do not
need the full fragment of Higher-Order Datalog$^\neg$. In this section, we
explore whether the same expressive power can be achieved without the use of
such variables but utilizing other powerful constructs of the language, namely
partially applied predicates. To this end, we introduce a semantics-preserving
transformation that converts every $(k+1)$-Order Datalog$^\neg$ program
containing existential predicate variables of order $k\geq 1$ into an equivalent
program of the same order that does not contain existential predicate variables.
This result implies that $(k+1)$-Order Datalog$^\neg$ without existential
predicate variables has the same expressive power as the full language.
Furthermore, the transformation we propose preserves stratification: if the
original program is stratified, then the transformed program is stratified as
well. In other words, stratified programs without existential predicate
variables can be as expressive as unstratified programs with existential
predicate variables. We first illustrate the proposed transformation with our
Hamilton example.
\begin{example}
Consider our initial rule for the Hamilton query:
\begin{lstlisting}
hamilton(X,Y):-ordering(Ord),first(Ord,X),last(Ord,Y),subset(succ(Ord),e).
\end{lstlisting}
The key idea of the transformation is that instead of using an existential
predicate variable for finding an appropriate relation \lstinline`Ord` that is a
strict total order, we can use an iterative procedure that starts from the empty
relation and successively adds to it pairs of individual constants until we get
a relation that is indeed a strict total order. In other words, we construct
ourselves, in a bottom-up way,  the strict total order. The corresponding transformed
program is the following:
\begin{lstlisting}
hamilton(X,Y):-test(X,Y,empty).
test(X,Y,Ord):-ordering(Ord),first(Ord,X),last(Ord,Y),subset(succ(Ord),e).
test(X,Y,Ord):-test(X,Y,add(Ord,Z$_1$,Z$_2$)).
\end{lstlisting}
%
%
Notice that in the last rule above, we add a new pair \lstinline`(Z$_1$,Z$_2$)` to the
\lstinline`Ord` relation; the definition of \lstinline`add` is quite simple and
will be given later in the section. The variables \lstinline`Z$_1$` and
\lstinline`Z$_2$` are existential but of lower order than the relation
\lstinline`Ord`. In other words, our transformation decreases by one the order
of the existential predicate variables that the program contains. Thus, if we
repeat this process successively, at the end we get a program that only contains
existential variables of type $\iota$. In the Hamilton program, one step
suffices to complete the transformation (the variables \lstinline`Z$_1$` and
\lstinline`Z$_2$` are of type $\iota$).%
\hfill\hbox{\proofbox}
\end{example}

We can now provide a more general description of the aforementioned transformation.
Let $\Prog$ be a program and assume it contains the following rule:
\begin{lstlisting}
p($\bar{\tt X}$):-$\mathsf{B}[\bar{\tt X},{\tt R}]$.
\end{lstlisting}
where $\mathsf{B}[\bar{\tt X},{\tt R}]$ is an expression containing
the variables $\bar{\tt X} = {\tt X}_1,\ldots,{\tt X}_m$
that occur in the head of the rule and also a predicate variable ${\tt R}$
of order $k \geq 1$ that does not occur in the head and thus is existentially quantified.
Let $\rho_{\tt R}$ be the type of ${\tt R}$.
%
%
We replace the aforementioned rule with the following set of rules:
\begin{lstlisting}
p($\bar{\tt X}$):-test$_{\mathsf{B}}$($\bar{\tt X}$,empty$_{\rho_{\tt R}}$).
test$_{\mathsf{B}}$($\bar{\tt X}$,R):-$\mathsf{B}[\bar{\tt X},{\tt R}]$.
test$_{\mathsf{B}}$($\bar{\tt X}$,R):-test$_{\mathsf{B}}$($\bar{\tt X}$,add$_{\rho_{\tt R}}$(R,$\bar{\tt Z}$)).
\end{lstlisting}
The predicate ${\tt empty}_{\rho_{\tt R}}$ defines the empty relation of objects
of type $\rho_{\tt R}$ and the predicate ${\tt add}_{\rho_{\tt R}}$ adds an
element to a relation of type $\rho_{\tt R}$.
The predicate ${\tt eq}_{\rho_{\bar{\tt Z}}}$ is a higher-order
equality predicate\ifincludeappendix (see Appendix~\ref{appendixC})\fi.
\begin{lstlisting}
empty$_{\rho_{\tt R}}$($\bar{\tt Y}$):-false.
add$_{\rho_{\tt R}}$(R,$\bar{\tt Z}$,$\bar{\tt Y}$):-R($\bar{\tt Y}$).
add$_{\rho_{\tt R}}$(R,$\bar{\tt Z}$,$\bar{\tt Y}$):-eq$_{\rho_{\bar{\tt Z}}}$($\bar{\tt Z}$,$\bar{\tt Y}$).
\end{lstlisting}
%
Note that \lstinline`add$_{\rho_{\tt R}}$(R,$\bar{\tt Z}$)` denotes a relation
that contains every element $\bar{\tt Y}$ of ${\tt R}$ and also $\bar{\tt Z}$.
The process introduces only the existential variables \lstinline`$\bar{\tt Z}$`,
which are of order at most $k-1$.

%
%
\ifincludeappendix
We have the following theorem, whose proof is given in Appendix~\ref{appendixC}.
\else
We have the following theorem, whose proof is given in the supplementary material.
\fi
\begin{theorem}\label{transformation-correctness-query}
Let $\Prog$ be a Higher-Order Datalog$^\neg$ program that defines a query
$\mathcal{Q}$ under
the well-founded semantics (resp. stable model semantics with cautious reasoning, stable model semantics with brave reasoning).
Let $\prog{P}'$ be the program that results by applying the aforementioned
transformation to some rule of $\Prog$. Then, $\prog{P}'$ defines the same query $\mathcal{Q}$
under the well-founded semantics (resp. stable model semantics with cautious reasoning, stable model semantics with brave reasoning).
\end{theorem}

By applying the transformation described above
to each rule and every existential variable of order $k$, we obtain a program
without such variables. Repeating this process iteratively for
variables of order $k-1$, $k-2$, and so on, we can eventually eliminate all
existential predicate variables from the initial program.


If we denote with Higher-Order Datalog$^{\neg,\not\exists}$ the fragment of
Higher-Order Datalog$^{\neg}$ that does not contain existential predicate
variables, then the following corollary is immediate:
\begin{corollary}
$(k+1)$-Order Datalog$^{\neg,\not\exists}$ and Stratified $(k+1)$-Order
Datalog$^{\neg,\not\exists}$ capture \EXPTIME[k] under the well-founded
semantics. $(k+1)$-Order Datalog$^{\neg,\not\exists}$ and Stratified+Choices
$(k+1)$-Order Datalog$^{\neg,\not\exists}$ under the stable model semantics
capture \coNEXPTIME[k] with cautious reasoning and \NEXPTIME[k] with brave
reasoning.
\end{corollary}
Obviously, Higher-Order Datalog$^{\neg,\not\exists}$ under any of the
aforementioned semantics captures \ELEMENTARY.

\section{Conclusions and Future Work\label{conclusions}}
We have presented an exploration of the expressive power of Higher-Order
Datalog$^\neg$ under the well-founded and the stable model semantics. Our
results identify fragments of the language that, despite being syntactically
restricted, possess the same expressive power as the full language. Moreover,
our results indicate that by increasing the order of programs under the
well-founded semantics we can surpass the expressive power of lower-order
programs under the stable model semantics.

There are several challenging directions for future work. First, although our
results indicate that by increasing the order of the programs the well-founded
semantics can match the power of stable model semantics, it is still unclear to
us if there exists a formal transformation from a $k$-order program $\mathsf{P}$ to
a $(k+1)$-order program $\mathsf{P}'$ such that the well-founded semantics of
$\mathsf{P}'$ captures, in some sense, the stable model semantics of
$\mathsf{P}$. Another direction that would be very interesting and certainly
quite challenging, would be the implementation of Higher-Order Datalog$^\neg$.
In general, implementing efficiently non-monotonic extensions of Datalog is
already non-trivial even at the first-order case. Probably, a promising
direction would be to identify interesting subclasses of Higher-Order
Datalog$^\neg$ that lend themselves to efficient implementation while at the
same time retaining some strong expressibility features of the language. For example,
as observed by one of the reviewers, it would be interesting to investigate
the notion of \emph{safety} in Higher-Order Datalog$^\neg$ and whether this notion
affects the expressiveness and the potential for efficient implementation of the language.

\bibliography{expressivity}

\ifincludeappendix
\clearpage
\appendix
\section{\hspace{-0.18cm}: The Semantics of Higher-Order Datalog$^\neg$}\label{semantics}

The semantics of the base type $\bool$ is the classical Boolean domain
$\{\mathit{true}, \mathit{false}\}$ and that of the base type $\basedom$ is
$U_{\mathsf{P}}$, namely the set of individual constant symbols in \Prog.
We start by defining the meanings of types of Higher-Order Datalog${^\neg}$.

\begin{definition}\label{def:orders_two-valued}
Let $\mathsf{P}$ be a Higher-Order Datalog${^\neg}$ program. We define the two-valued meaning $\mo{\tau}_{U_{\mathsf{P}}}$ and the
three-valued meaning $\mos{\tau}_{U_{\mathsf{P}}} \supseteq \mo{\tau}_{U_{\mathsf{P}}}$ of a
type $\tau$ with respect to $U_{\mathsf{P}}$, as follows:
\begin{itemize}
  \item $\mo{\bool}_{U_{\mathsf{P}}} = \{\mathit{true}, \mathit{false}\}$ and $\mos{\bool}_{U_{\mathsf{P}}} = \{\mfalse,\mundef,\mtrue\}$.
  The partial order $\leq_\bool$ is the one induced by the ordering $\mfalse <_\bool \mundef <_\bool \mtrue$;
  the partial order $\preceq_\bool$ is the one induced by the ordering $\mundef \prec_\bool \mfalse$ and $\mundef \prec_\bool \mtrue$.

  \item $\mo{\basedom}_{U_{\mathsf{P}}} = \mos{\basedom}_{U_{\mathsf{P}}} = U_{\mathsf{P}}$.
        The partial order $\leq_\basedom$ is the
        trivial one defined as $d \leq_\basedom d$ for all $d \in U_{\mathsf{P}}$.
  \item $\mo{\rho \to \pi}_{U_{\mathsf{P}}} = \mo{\rho}_{U_{\mathsf{P}}} \to \mo{\pi}_{U_{\mathsf{P}}}$,
        namely the set of all functions from  $\mo{\rho}_{U_{\mathsf{P}}}$ to $\mo{\pi}_{U_{\mathsf{P}}}$ and
        $\mos{\rho \to \pi}_{U_{\mathsf{P}}} = \mo{\rho}_{U_{\mathsf{P}}} \to \mos{\pi}_{U_{\mathsf{P}}}$, namely
        the set of all functions from $\mo{\rho}_{U_{\mathsf{P}}}$ to $\mos{\pi}_{U_{\mathsf{P}}}$.
        The partial order $\leq_{\rho \to \pi}$ is defined as follows:
        for all $f,g \in \mos{\rho \to \pi}_{U_{\mathsf{P}}}$,
        $f \leq_{\rho \to \pi} g$ iff $f(d) \leq_{\pi} g(d)$ for all $d \in \mo{\rho}_{U_{\mathsf{P}}}$.
        The partial order $\preceq_{\rho \to \pi}$ is defined as follows:
        for all $f,g \in \mos{\rho \to \pi}_{U_{\mathsf{P}}}$,
        $f \preceq_{\rho \to \pi} g$ iff $f(d) \preceq_{\pi} g(d)$ for all $d \in \mo{\rho}_{U_{\mathsf{P}}}$.
\end{itemize}
\end{definition}



%
The subscripts from the above partial orders will be omitted when they are
obvious from context. Moreover, we will omit the subscript $U_{\mathsf{P}}$
assuming that our semantics is defined with respect to a specific program
$\mathsf{P}$. 
For every predicate type $\pi$, $(\mo{\pi}, \leq_\pi)$ is a complete lattice
and $(\mo{\pi}, \preceq_\pi)$ is a complete meet-semilattice. 
We denote by $\bigvee_{\leq_{\pi}}$ and
$\bigwedge_{\leq_{\pi}}$ the corresponding lub and glb operations of the above
lattice and by $\bigwedge_{\preceq_{\pi}}$ the corresponding glb operation of the above 
semilattice.

To denote that an expression $\mathsf{E}$ has type $\rho$ we will often write $\mathsf{E}:\rho$.

\begin{definition}\label{def:interpretation_Herbrand}
A three-valued Herbrand interpretation $I$
of a program $\mathsf{P}$ assigns to each
individual constant $\mathsf{c}$ of $\mathsf{P}$, the element
$I(\mathsf{c}) = \mathsf{c}$, and to each predicate constant
$\mathsf{p} : \pi$ of $\mathsf{P}$, an element $I(\mathsf{p}) \in \mos{\pi}$.
An interpretation is called two-valued if for every $\mathsf{p} : \pi$ of $\mathsf{P}$, $I(\mathsf{p}) \in \mo{\pi}$.
\end{definition}

We will denote the set of three-valued Herbrand interpretations of a program $\mathsf{P}$
with $\mathcal{H}_\mathsf{P}$ and the set of two-valued Herbrand interpretations with
$H_\mathsf{P}$. Since $\mo{\pi} \subseteq \mos{\pi}$ for all types $\pi$ it follows also that
$H_\mathsf{P} \subseteq \mathcal{H}_\mathsf{P}$.
We define a partial order on $\mathcal{H}_\mathsf{P}$ as
follows: for all $I, J \in \mathcal{H}_\mathsf{P}$, $I \leq J$
iff for every predicate constant $\mathsf{p} : \pi$ that appears in
$\mathsf{P}$, $I(\mathsf{p}) \aleq[\pi] J(\mathsf{p})$.

\begin{definition}\label{def:state_Herbrand}
A \emph{Herbrand state} $s$ of a program $\mathsf{P}$ is a function that
assigns to each argument variable $\mathsf{R}$ of type $\rho$, an element
$s(\mathsf{R}) \in \mo{\rho}$.  We denote the set of Herbrand states with
$S_\mathsf{P}$.
\end{definition}
In the following, $s[\mathsf{R}_1/d_1,\ldots,\mathsf{R}_n/d_n]$ is used to
denote a state that is identical to $s$ the only difference being that the new
state assigns to each $\mathsf{R}_i$ the corresponding value $d_i$; for brevity,
we will also denote it by $s[\bar{\mathsf{R}}/\bar{d}]$.
\enlargethispage{1\baselineskip}

\begin{definition}\label{tuple-semantics}
Let $\mathsf{P}$ be a Higher-Order Datalog${^\neg}$ program, $\mathcal{I}$ a three-valued Herbrand
interpretation of $\mathsf{P}$, and $s$ a Herbrand state.
The \emph{three-valued semantics} of expressions and bodies is defined as follows:
\begin{enumerate}
  \item $\mwrst{\mathsf{R}}{\mathcal{I}}{s} = s(\mathsf{R})$
  \item $\mwrst{\mathsf{c}}{\mathcal{I}}{s} = \mathcal{I}(\mathsf{c}) = \mathsf{c}$
  \item $\mwrst{\mathsf{p}}{\mathcal{I}}{s} = \mathcal{I}(\mathsf{p})$
  \item \label{item:threeval-apply}$\mwrst{(\mathsf{E}_1\ \mathsf{E}_2)}{\mathcal{I}}{s} = \bigwedge_{\preceq_{\pi}}\{\lsem \mathsf{E}_1 \rsem^{*}_s(\mathcal{I})(d) \mid d \in \lsem \rho\rsem, \lsem \mathsf{E}_2 \rsem^{*}_s(\mathcal{I}) \preceq_{\rho} d\}$,
      for $\mathsf{E}_1\! :\! \rho \to \pi$ and $\mathsf{E}_2\! :\! \rho$
  \item $\mwrst{(\mathsf{E}_1\approx \mathsf{E}_2)}{\mathcal{I}}{s} = \begin{cases}
    \mathit{true},  & \text{if } \mwrst{\mathsf{E}_1}{\mathcal{I}}{s} = \mwrst{\mathsf{E}_2}{\mathcal{I}}{s} \\
    \mathit{false}, & \text{otherwise}
    \end{cases}$
  \item $\mwrst{(\sim \mathsf{E})}{\mathcal{I}}{s} =  (\mwrst{\mathsf{E}}{\mathcal{I}}{s})^{-1}$, with $\mathit{true}^{-1}\!=\!\mathit{false}$, $\mathit{false}^{-1}\!=\!\mathit{true}$
  and $\mathit{undef}^{-1}\!=\!\mathit{undef}$
  \item $\mwrst{(\mathsf{E}_1 \wedge \cdots \wedge \mathsf{E}_m)}{\mathcal{I}}{s} =
    \bigwedge_{\leq_\bool}\{\mwrst{\mathsf{E}_1}{\mathcal{I}}{s},\ldots,\mwrst{\mathsf{E}_m}{\mathcal{I}}{s}\}$
\end{enumerate}
\end{definition}

\begin{definition}\label{def:three-valued_model}
Let $\mathsf{P}$ be a program and $\mathcal{M}$ be a three-valued Herbrand
interpretation of $\mathsf{P}$. Then, $\mathcal{M}$ is a
\emph{three-valued Herbrand model} of $\mathsf{P}$ iff for every rule
$\mathsf{p}\ \overline{\mathsf{R}} \lrule \mathsf{B}$ in $\mathsf{P}$
and for every Herbrand state $s$,
$\mwrst{\mathsf{B}}{\mathcal{M}}{s} \leq_\bool \mwrst{\mathsf{p}\ \bar{\mathsf{R}}}{\mathcal{M}}{s}$.
\end{definition}

\begin{definition}\label{def:three-valuedTP}
  Let $\mathsf{P}$ be a Higher-Order Datalog${^\neg}$ program. The \emph{three-valued immediate consequence operator}
  $\mathcal{T}_\mathsf{P} : \mathcal{H}_\mathsf{P} \to \mathcal{H}_\mathsf{P}$ is defined
  for every predicate constant $\mathsf{p} : \rho_1 \to \cdots \to \rho_n \to \bool$ in $\mathsf{P}$ and
  all $d_1 \in \mo{\rho_1},\ldots, d_n \in \mo{\rho_n}$, as:
  $\mathcal{T}_{\mathsf{P}}(\mathcal{I})(\mathsf{p})(\bar{d}) =
        \bigvee_{\leq_\bool}\{
          \mwrst{\mathsf{B}}{\mathcal{I}}{s[\bar{\mathsf{R}}/\bar{d}]} \mid
                  \mbox{$s\in S_{\mathsf{P}}$ and
                        $(\mathsf{p}\ \bar{\mathsf{R}} \lrule \mathsf{B})$ in $\mathsf{P}$}\}$.
\end{definition}


\begin{definition}\label{def:consistent_lattice}\label{orderings_on_pairs}
Let $(L,\leq)$ be a complete lattice. We define $L^{c} =\{(x,y) \in L \times L \mid x \leq y\}$.
Moreover, we define the relations $\leq$ and  $\preceq$, so that
for all $(x,y),(x',y') \in L^c$:
%
$(x, y) \leq (x', y')$ iff $x \leq x'$ and $y \leq y'$, and
$(x, y) \preceq (x', y')$ iff $x \leq x'$ and $y' \leq y$.
%
\end{definition}

We will denote the \emph{first} and \emph{second} selection
functions on pairs with the more compact notation $[\cdot]_1$ and $[\cdot]_2$:
given any pair $(x,y)$, it is $[(x,y)]_1 = x$ and $[(x,y)]_2 = y$.
It is easy to see that $L^{c}$ is a complete lattice with respect to $\leq$ where $\bigvee_{\leq}$
and $\bigwedge_{\leq}$ are defined in a pointwise way and $L^{c}$ is a meet-semilattice with respect to $\preceq$
where $\bigwedge_{\preceq} S = (\bigwedge\{ [x]_1 \mid x \in S \}, \bigvee \{ [x]_2 \mid x \in S \})$.

\begin{definition}\label{def:pair-semantics}
Let $\mathsf{P}$ be a program, $(I, J) \in H^c_\mathsf{P}$, and $s \in S_{\mathsf{P}}$.
The pair semantics of expressions and bodies is defined as follows:
\begin{enumerate}
  \item $\mwrc{\mathsf{R}}{I,J}{s} = (s(\mathsf{R}),s(\mathsf{R}))$
  \item $\mwrc{\mathsf{c}}{I,J}{s} = (I(\mathsf{c}), J(\mathsf{c}))$
  \item $\mwrc{\mathsf{p}}{I,J}{s} = (I(\mathsf{p}), J(\mathsf{p}))$
  \item $\mwrc{(\mathsf{E}_1\ \mathsf{E}_2)}{I,J}{s} = (\bigwedge_{\leq_\pi} \{f(d) \mid d \in \lsem \rho\rsem, l \leq d \leq u \}, \bigvee_{\leq_\pi} \{g(d) \mid d \in \lsem \rho\rsem, l \leq d \leq u\})$,
      where $(f,g) = \mwrc{\mathsf{E}_1}{I,J}{s}$, $(l,u) = \mwrc{\mathsf{E}_2}{I,J}{s}$ for $\mathsf{E}_1\! :\! \rho \to \pi$ and $\mathsf{E}_2\! :\! \rho$.
  \item $\mwrc{(\mathsf{E}_1\approx \mathsf{E}_2)}{I,J}{s} = \begin{cases}
    (\mathit{true}, \mathit{true}),  & \text{if } \mwrc{\mathsf{E}_1}{I,J}{s} = \mwrc{\mathsf{E}_2}{I,J}{s} \\
    (\mathit{false}, \mathit{false}), & \text{otherwise}
    \end{cases}$
  \item $\mwrc{(\sim \mathsf{E})}{I,J}{s} =  \mwrc{\mathsf{E}}{I,J}{s}^{-1}$, with $(\mathit{true},\mathit{true})^{-1}\!=\!(\mathit{false},\mathit{false})$, $(\mathit{false},\mathit{false})^{-1}\!=\!(\mathit{true},\mathit{true})$
  and $(\mathit{false},\mathit{true})^{-1}\!=\!(\mathit{false},\mathit{true})$
  \item $\mwrc{(\mathsf{E}_1 \wedge \cdots \wedge \mathsf{E}_m)}{I,J}{s} =
  \bigwedge_{\leq_\bool}\{\mwrc{\mathsf{E}_1}{I,J}{s},\ldots,\mwrc{\mathsf{E}_m}{I,J}{s}\}$
\end{enumerate}
\end{definition}

\begin{definition}\label{def:pair-ap}
Let $\mathsf{P}$ be a program. The \emph{approximating operator}
$\ATP : H^c_\mathsf{P} \to H^c_\mathsf{P}$ is defined
for every predicate constant $\mathsf{p} : \rho_1 \to \cdots \to \rho_n \to \bool$ in $\mathsf{P}$ and
all $d_1 \in \mo{\rho_1},\ldots, d_n \in \mo{\rho_n}$, as $\ATP(I,J) = (\ATP(I,J)_1, \ATP(I,J)_2)$
where, for $i \in \{ 1, 2 \}$:
\[ \ATP(I,J)_i(\mathsf{p}) (\bar{d})= \bigvee\nolimits_{\leq_\bool}\{
  [\mwrc{\mathsf{B}}{I,J}{s[\bar{\mathsf{R}}/\bar{d}]}]_i \mid
          \mbox{$s\in S_{\mathsf{P}}$ and
                $(\mathsf{p}\ \bar{\mathsf{R}} \lrule \mathsf{B})$ in $\mathsf{P}$}\}
\]
\end{definition}

We proceed with a definition of an isomorphism,
\ie an order-preserving bijection,
for every predicate type $\pi$, between the set $\mos{\pi}$ of three-valued meanings and the set $\cp{\mo{\pi}}$ of  pairs of two-valued ones.

\begin{definition}\label{tau-definition}
For every predicate type $\pi$, we define the functions
$\tau_\pi: \mos{\pi} \to \cp{\mo{\pi}}$ and
$\tau^{-1}_\pi: \cp{\mo{\pi}} \to \mos{\pi}$,
as follows:
\begin{itemize}
  \item $\tau_o(\mfalse) = (\mfalse, \mfalse)$, $\tau_o(\mtrue) = (\mtrue, \mtrue)$,
                        $\tau_o(\mundef) = (\mfalse, \mtrue)$
  \item $\tau_{\rho \to \pi}(f) = (\lambda d. [\tau_\pi(f(d))]_1, \lambda d. [\tau_\pi(f(d))]_2)$

\end{itemize}
and
\begin{itemize}
\item $\tau^{-1}_o(\mfalse,\mfalse) = \mfalse$, $\tau^{-1}_o(\mtrue,\mtrue) = \mtrue$,
                              $\tau^{-1}_o(\mfalse, \mtrue) = \mundef$
\item $\tau^{-1}_{\rho \to \pi}(f_1, f_2) = \lambda d. \tau^{-1}_{\pi}(f_1(d), f_2(d))$
\end{itemize}
\end{definition}

The functions $\tau_\pi$ defined above, can easily be extended to functions
between ${\cal H}_{\mathsf{P}}$ and $H^{c}_\mathsf{P}$: given ${\cal I} \in
{\cal H}_{\mathsf{P}}$, we define $\tau({\cal I}) = (I,J)$, where for every
predicate constant $\mathsf{p}:\pi$ it holds $I(\mathsf{p}) = [\tau_{\pi}({\cal
I}(\mathsf{p}))]_1$ and $J(\mathsf{p}) = [\tau_{\pi}({\cal I}(\mathsf{p}))]_2$.
Conversely, given a pair $(I,J) \in H^c_\mathsf{P}$, we define the
three-valued Herbrand interpretation $\tau^{-1}(I,J)$, for every predicate
constant $\mathsf{p}:\pi$, as follows: $\tau^{-1}(I,J)(\mathsf{p}) =
\tau^{-1}_{\pi}(I(\mathsf{p}),J(\mathsf{p}))$.

\enlargethispage{1\baselineskip}
\begin{proposition}
Let $\mathsf{P}$ be a Higher-Order Datalog${^\neg}$ program. Then,
$\ATP(I,J) = \tau(\mathcal{T}_\mathsf{P}(\tau^{-1}(I,J)))$.
\end{proposition}
\begin{proof}
Since $\tau$ preserves the orderings $\leq,\preceq$ \citep{iclp24}, it is easy to show it commutes with greatest lower bounds/least upper bounds
 taken with respect to $\leq$
and greatest lower bounds with respect to $\preceq$.
The result follows easily by showing that 
for any $\mathcal{I}$ and $s$, $\mwrc{\mathsf{E}}{\tau(\mathcal{I})}{s}=\tau(\mwrst{\mathsf{E}}{\mathcal{I}}{s})$
using induction on the structure of $\mathsf{E}$.
\end{proof}

In the following proofs, we will be using Definition~\ref{def:pair-semantics} 
for the semantics for expressions and $\AP$ of Definition~\ref{def:pair-ap}.
\begin{definition}\label{def:AFTsemantics}
Let $\mathsf{P}$ be a Higher-Order Datalog${^\neg}$ program.
We call:
\begin{itemize}
  \item $\tau^{-1}(I,J)$ a \emph{three-valued stable model of $\mathsf{P}$} if $(I,J)$ is a stable fixpoint of $\ATP$;
        that is, if $I=\lfp \ATP(\cdot,J)_1$ and $J=\lfp \ATP(I,\cdot)_2$;
  \item $\tau^{-1}(I,J)$ a \emph{stable model of $\mathsf{P}$} if $(I,J)$ is a stable fixpoint of $\ATP$ and $I=J$;
  \item $\tau^{-1}(I,J)$ is the \emph{well-founded model of $\mathsf{P}$} if 
        it is the $\preceq$-least three-valued stable model of $P$.
\end{itemize}
\end{definition}

\section{\hspace{-0.18cm}: Complexity of Computing the Semantics}\label{appendixB}

In this appendix we show an upper bound in the expressive power of the well-founded
semantics and the stable model semantics of $(k+1)$-Order Datalog$^\neg$.

We will use the expression $\exp_{k}(poly(n))$ and denote with $poly(n)$
a fixed polynomial of $n$ to hide when possible the different constants from
the analysis.



\begin{retheorem}{turing-machine-well-founded}
Let $\Prog$ be a $(k+1)$-Order Datalog$^\neg$ program that defines a query
$\mathcal{Q}_\mathsf{P}$ under the well-founded semantics.
Then, there exists a deterministic Turing
machine that takes as input an encoding of a database $D$
that uses $n$ individual constant symbols and
a ground atom $p(\bar{a})$, where $p$ is a predicate constant of $\Prog$ and
$\bar{a}$ is a tuple of those individual constants, and
decides whether $p(\bar{a}) \in \mathcal{Q}_\mathsf{P}(D)$, in at most
$exp_{k}(n^d)$ steps for some constant $d$.
\end{retheorem}

\begin{proof}
In the following, we assume that the representation of the input database $D$ is sensible, \ie
each of the $n$ different \emph{individual} constants is represented
with $O(\log n)$ bits and the input does not contain multiple entries or irrelevant information. Since we consider generic queries,
a compact representation of the constants is desired in order to argue about the true complexity.
Furthermore, we expect the input database to belong to the input schema. Therefore, we consider the relation symbols as fixed.
We will describe a Turing Machine that computes the well-founded model of $\Prog \cup D$
and accepts if and only if $p(\bar{a})$ is true in the well-founded model. We will argue that the machine does so in
$exp_{k}(n^d)$ steps for some fixed constant $d$.
We will demonstrate the following claims
that bound the complexity of the computation:
\begin{enumerate}
\item For every expression in the program that is of type
      $\rho$ that is at most $k$-order, the total number of elements in
      $\mo{\rho}_{U_{\mathsf{P}\cup D}}$ is at most $\exp_{k}(poly(n))$ and the
      total space we require to represent each in the machine's tape is
      $\exp_{k-1}(poly(n))$ bits or cells, where $\exp_{-1}(poly(n))$ denotes by convention $O(\log n)$.
      The complexity of comparing two such elements
      is $\exp_{k-1}(poly(n))$.
\item The size of a two-valued interpretation is at most $\exp_k(poly(n))$.
\item The application of the operator $A_{\mathsf{P}\cup D}$ can be performed in $\exp_{k}(poly(n))$ steps.
\item The total number of applications of $A_{\mathsf{P}\cup D}$ required to
      compute the well-founded fixpoint is at most $\exp_{k}(poly(n))$.
\end{enumerate}

We now argue about each claim specifically.

\paragraph{Claim 1.}

Let $\rho$ be a type with $order(\rho)=j$.
We show using induction on the structure of $\rho$ that
$|\mo{\rho}_{U_{\mathsf{P}\cup D}}|\leq\exp_{j}(poly(n))$ and that 
for all $x \in \mo{\rho}_{U_{\mathsf{P}\cup D}}$ we can represent $x$ in memory with at most
$\exp_{j-1}(poly(n))$ cells. The base case for types $\iota$ and $o$ follows immediately from our
assumptions about the representation of the input. Assume without loss of generality that $\rho=
\rho_1 \rightarrow \cdots \rightarrow \rho_m \rightarrow o$ and we have that for
each $\rho_i$ it is $order(\rho_i)\leq j-1$. We
represent any $x \in \mo{\rho}_{U_{\mathsf{P}\cup D}}$ in memory as a set of objects.
That is, we represent it as a series of entries of the form
$(x_1,x_2,\ldots,x_m)$ where $x_i \in \mo{\rho_i}_{U_{\mathsf{P}\cup D}}$. By induction hypothesis,
each $d_i$ has a size of at most $\exp_{j-2}(poly(n))$ bits
and there are at most $\exp_{j-1}(poly(n))$ elements for that subtype. Therefore,
the size of each entry is at most $m\cdot\exp_{j-2}(poly(n))$ bits and
there are at most $(\exp_{j-1}(poly(n))^m$ such entries. The total
size of the representation is thus bounded by
$m\cdot\exp_{j-2}(poly(n))\cdot (\exp_{j-1}(poly(n))^m$ so bounded by
$\exp_{j-1}(poly(n))$ bits. In order to argue about the cardinality for that
type it suffices to notice that we can produce each element of
$\mo{\rho}_{U_{\mathsf{P}\cup D}}$ by choosing a different possible subset of entries.
Therefore, there are at most $2^{(\exp_{j-1}(poly(n))^m}$ elements or, alternatively,
at most $\exp_{j}(poly(n))$.
Finally, for any $x,x' \in \mo{\rho}_{U_{\mathsf{P}\cup D}}$ the time needed to
check whether $x \leq x'$ is at most $\exp_{j-1}(poly(n))$. This is
easily shown again using induction on the structure of the types. To compare two
elements we have to check at most $\exp_{j-1}(poly(n))$ entries and verify
that every entry in the representation of $x$ also belongs to the
representation of $x'$. The cost of comparing two entries is
$m\cdot\exp_{j-2}(poly(n))$ by induction hypothesis and therefore it follows that
the total running time is controlled by the amount of entries and can be performed in
at most $\exp_{j-1}(poly(n))$ steps.

\paragraph{Claim 2.}
Since each constant predicate of $P$ is at most of $k+1$ order and there is a fixed
number of those, the total size of a two-valued interpretation $I$ is at most
$\exp_{k}(poly(n))$ bits. It is important to notice that we do not have to
produce all the elements of a type of order $k+1$ at any point in the computation.

\paragraph{Claim 3.}
$A_{\mathsf{P}\cup D}$ works on pairs of two-valued interpretations. We show that a
new pair of interpretations $(I',J')=A_{\mathsf{P}\cup D}(I,J)$ can be computed in $\exp_{k}(poly(n))$ steps.

To do that we bound the computation time for computing a (fixed) expression
$\mathsf{E}$ of the $(k+1)$-Order Datalog$^\neg$ program. Specifically, we show that computing
$\mwrc{\mathsf{E}}{I,J}{s}$ for some $I,J$, $I \leq J$ takes at most $\exp_{k}(poly(n))$ steps by
using induction on the structure of the expression $\mathsf{E}$. For the base case, where
$\mathsf{E}$ is either constant or variable, it follows  by the size of the representation of the program types. Most other cases
also follow easily. The only case that needs further justification, is the case of expression application.

Let $\mathsf{E} := (\mathsf{E}_1\ \mathsf{E}_2)$. By induction hypothesis, we have that both
$(f_1,f_2)=\mwrc{\mathsf{E}_1}{I,J}{s}$ and $(g_1,g_2)=\mwrc{\mathsf{E}_2}{I,J}{s}$
can be computed within $\exp_{k}(poly(n))$ steps. By following the semantic
definition of application, we have that:
$\mwrc{(\mathsf{E}_1\ \mathsf{E}_2)}{I,J}{s} = (\bigwedge\{f_1(x)\mid g_1 \leq x \leq g_2\},\bigvee
\{f_2(x)\mid g_1 \leq x \leq g_2\})$. Let $\rho$ be the type of $\mathsf{E}_2$ and $x\in \lsem \rho \rsem_{U_{\mathsf{P}\cup D}}$,
we need to compute $f_1(x)$. Remember that the result of the application is another set (or a boolean value).
To perform this we need to collect all entries of set $f_1$ whose first element is $x$
and create a new set with only them, after we erase $x$ from the beginning of each such entry thus resulting in the new set.
Recalling the arguments analyzed previously, this can be done in at most
$\exp_{k}(poly(n))$ steps since this is a bound on the number of entries in
$f$ and each entry manipulation can be done in $\exp_{k-1}(poly(n))$ steps.
The whole process including storing the result can be done in
$\exp_{k}(poly(n))$ steps. To compute $\bigwedge\{f_1(x)\mid g_1 \leq x \leq g_2\}$ we need to perform this step at most
$|\mo{\rho}_{U_{\mathsf{P}\cup D}}|$ times since $\{x \mid g_1 \leq x \leq g_2\} \subseteq
\mo{\rho}_{U_{\mathsf{P}\cup D}}$. But since $\mathsf{E}_2$ has a type that appears as
argument, it is of at most $k$-order thus we have to perform the application step
at most $\exp_{k}(poly(n))$ times. Finally, it is also easy to see that taking the glb of
the former results can also be done in $\exp_{k}(poly(n))$ steps and it follows that the whole
computation of the result can be performed in at most $\exp_{k}(poly(n))$ steps.
Likewise, we argue about the second element of the pair.

Next, for some predicate ${\tt p}$ to compute $A_{\mathsf{P}\cup D}(I,J)_i({\tt p})$ we need to
compute for every argument $\bar{d}$ the value of $A_{\mathsf{P}\cup D}(I,J)_i({\tt p})(\bar{d})$.
Recall that  \[ A_{\mathsf{P}\cup D}(I,J)_i(\mathsf{p}) (\bar{d})= \bigvee\nolimits_{\leq_\bool}\{
  [\mwrc{\mathsf{B}}{I,J}{s[\bar{\mathsf{R}}/\bar{d}]}]_i \mid
          \mbox{$s\in S_{\mathsf{P}\cup D}$ and
                $(\mathsf{p}\ \bar{\mathsf{R}} \lrule \mathsf{B})$ in $\mathsf{P} \cup D$}\}
\]
For each $\bar{d}$ a fixed number of bodies $\mwrc{\mathsf{B}}{I,J}{s[\bar{\mathsf{R}}/\bar{d}]}$
need to be computed for every $s\in S_{\mathsf{P}\cup D}$ that assigns $\bar{d}$ to the formal variables.
Therefore, we require $\exp_{k}(poly(n))$ steps to compute each body for such $s$ and we require to do it at most
$\exp_{k}(poly(n))$ times since every existential free variable in the bodies is of order at most $k$
and the domain of the types for such variables ranges over $\exp_{k}(poly(n))$ elements. Finally, for the whole representation of
${\tt p}$ we need to perform this for every possible $\bar{d}$ and there are
$\exp_{k}(poly(n))$ many of them since each $d_i$ in tuple $\bar{d}$ belongs to a type of order
at most $k$. Clearly the pair interpretation with respect to ${\tt p}$ of $A_{\mathsf{P}\cup D}(I,J)(\mathsf{p})$ can be
computed in at most  $\exp_{k}(poly(n))$ steps and that also holds for the
total $A_{\mathsf{P}\cup D}(I,J)$ since the number of different predicates is fixed for fixed programs and schemas of input databases.

\paragraph{Claim 4.}
The well-founded fixpoint can be computed by the following sequence of
interpretations $\{(I_{m},J_{m})\}$:
\begin{align*}
(I_0,J_0) & =(\bot,\top) \\
(I_{m+1},J_{m+1}) & =(\lfp A_{\mathsf{P}\cup D}(\cdot,J_m)_1,\lfp A_{\mathsf{P}\cup D}(I_m,\cdot)_2)
\end{align*}
To compute $\lfp A_{\mathsf{P}\cup D}(\cdot,J_m)_1$ we start from $\bot$ and iterate the
operator $A_{\mathsf{P}\cup D}(\cdot,J_m)_1$. Each predicate in any interpretation $I$ viewed as a
set contains at most $\exp_{k}(poly(n))$ elements. Each iteration adds at
least one new object/entry to at least one predicate/set unless convergence has been achieved. Therefore,
we require at most $\exp_{k}(poly(n))$ iterations. Similarly,
for $\lfp A_{\mathsf{P}\cup D}(I_m,\cdot)_2$ which we compute starting from $I_m$ and then iterating
the operator $A_{\mathsf{P}\cup D}(I_m,\cdot)_2$.

Since $I_{m+1} \geq I_m$ and  $J_{m+1} \leq J_m$ we have that for at least one
predicate  either its $I$ interpretation increased or its $J$ interpretation
decreased, or we converged. Therefore, the convergence for the fixpoint of the
sequence $\{(I_{m},J_{m})\}$ also happens in at most $\exp_{k}(poly(n))$
iterations. We conclude that the total number of times we will have to apply
$A_{\mathsf{P}\cup D}$ is at most $\exp_{k}(poly(n)) \cdot 2 \cdot \exp_{k}(poly(n))$.
Considering that each application of the operators takes $\exp_{k}(poly(n))$
computational steps as shown above, we retrieve the desired total complexity.
\end{proof}




\begin{retheorem}{Turing-machine-stable-model}
Let $\Prog$ be a $(k+1)$-Order Datalog$^\neg$ program that defines a query
$\mathcal{Q}_\mathsf{P}$ under the stable model semantics and cautious reasoning.
Then, there exists a non-deterministic Turing
machine that takes as input an encoding of a database $D$
that uses $n$ individual constant symbols and
a ground atom $p(\bar{a})$, where $p$ is a predicate constant of $\Prog$ and
$\bar{a}$ is a tuple of those individual constants, and
decides whether $p(\bar{a}) \not\in \mathcal{Q}_\mathsf{P}(D)$, in at most
$exp_{k}(n^d)$ steps for some constant $d$.
\end{retheorem}

\begin{proof}
We can create such a non-deterministic Turing machine $M$ that operates as follows: $M$
guesses an interpretation $I$ by instantiating all elements of each type of order at most $k$
and then choosing in a non-deterministic way whether an entry of elements of appropriate type belongs to a predicate set or not.
In essence, the machine creates a possible two-valued interpretation of  $\mathsf{P}\cup D$.
Given the size of the representation of each predicate, it is clear that the guessing
can be done in $exp_{k}(poly(n))$ steps for some fixed polynomial of $n$. Then the machine
checks whether $I$ is a stable model by performing one internal loop similarly to that of the proof of Theorem~\ref{turing-machine-well-founded}.
As we have already argued, this step also takes
$exp_{k}(poly(n))$ steps at most. Finally, the machine accepts if $I$ is a stable model and
$\bar{a}$ does not belong to the set that denotes predicate $p$ in $I$.
Otherwise, the machine rejects for this computation path. It is clear that
$p(\bar{a}) \not\in \mathcal{Q}_\mathsf{P}(D)$ under stable model semantics with cautious
reasoning if and only if there exists a stable model where $p(\bar{a})$ is not in the output
relation if and only if the non-deterministic machine described here accepts.
\end{proof} 
\section{\hspace{-0.18cm}: Proofs of Section~\ref{transformation}}\label{appendixC}

In this section we establish the correctness of the transformation described in Section~\ref{transformation}.

Let $\Prog$ be a $(k+1)$-Order Datalog$^\neg$ program, \lstinline`p` be a predicate constant in $\Prog$
and \lstinline`p($\bar{\tt X}$):-$\mathsf{B}_i$` for $i = 1, \ldots, m$ be the rules of \lstinline`p` in $\Prog$.
We will write $\mathsf{B}_i[{\tt X}_1, \ldots {\tt X}_n]$ to denote, if necessary,
that variables ${\tt X}_1, \ldots {\tt X}_n$ occur in the body $\mathsf{B}_i$.
Let \lstinline`p($\bar{\tt X}$):-$\mathsf{B}_i[{\tt R}]$` be the $i$-th rule of \lstinline`p`,
\lstinline`R` be a $k$-order predicate variable of type $\rho_{\tt R} = \rho_{{\tt Z}_1} \to \cdots \to \rho_{{\tt Z}_n} \to o$
and ${\tt R} \not\in \bar{\tt X}$.\footnote{%
Existential variable ${\tt R}$ of type $o$ can be easily treated in a unified manner by replacing every occurrence of 
${\tt R}$ in $\mathsf{B}_i[{\tt R}]$ with ${\tt R(W)}$ where ${\tt R}$ now is an existential 
variable of type $\iota \to o$ and ${\tt W}$ is a fresh existential variable of type $\basedom$.
}
The transformation produces a program $\prog{P}'$ that contains
all the rules of $\Prog$ except \lstinline`p($\bar{\tt X}$):-$\mathsf{B}_i[{\tt R}]$` and
in addition it contains the following rules:
\begin{lstlisting}
p($\bar{\tt X}$):-test$_{\mathsf{B}_i}$($\bar{\tt X}$,empty$_{\rho_{\tt R}}$).
test$_{\mathsf{B}_i}$($\bar{\tt X}$,R):-$\mathsf{B}_i[{\tt R}]$.
test$_{\mathsf{B}_i}$($\bar{\tt X}$,R):-test$_{\mathsf{B}_i}$($\bar{\tt X}$,add$_{\rho_{\tt R}}$(R,$\bar{\tt Z}$)).
empty$_{\rho_{\tt R}}$($\bar{\tt Y}$):-false.
add$_{\rho_{\tt R}}$(R,${\tt Z}_1$,...,${\tt Z}_n$,${\tt Y}_1$,...,${\tt Y}_n$):-R(${\tt Y}_1$,...,${\tt Y}_n$).
add$_{\rho_{\tt R}}$(R,${\tt Z}_1$,...,${\tt Z}_n$,${\tt Y}_1$,...,${\tt Y}_n$):-eq$_{\rho_{{\tt Z}_1}}$(${\tt Z}_1$,${\tt Y}_1$),...,eq$_{\rho_{{\tt Z}_n}}$(${\tt Z}_n$,${\tt Y}_n$).
\end{lstlisting}
and for every type $\rho_{{\tt Z}_j}$, $j \in \{ 1, \ldots, n \}$ we also include the appropriate predicate {\tt eq}$_{\rho_{{\tt Z}_j}}$. If
$\rho_{{\tt Z}_j} = \iota$ then {\tt eq}$_{\rho_{{\tt Z}_j}}$ is defined as:
\begin{lstlisting}
eq$_{\rho_{{\tt Z}_j}}$(Z,Y):-Z=Y.
\end{lstlisting}
and if ${\rho_{{\tt Z}_j}} = \rho_1\to\cdots\to\rho_m \to o$, it is defined as:
\begin{lstlisting}
eq$_{\rho_{{\tt Z}_j}}$(Z,Y):-not neq$_{\rho_{{\tt Z}_j}}$(Z,Y).
neq$_{\rho_{{\tt Z}_j}}$(Z,Y):-Z(X$_1$,...,X$_m$),not Y(X$_1$,...,X$_m$).
neq$_{\rho_{{\tt Z}_j}}$(Z,Y):-not Z(X$_1$,...,X$_m$),Y(X$_1$,...,X$_m$).
\end{lstlisting}
We assume that the predicates
\lstinline`test$_{\mathsf{B}_i}$`,
\lstinline`add$_{\rho_{\tt R}}$`,
\lstinline`empty$_{\rho_{\tt R}}$`,
\lstinline`eq$_{\rho_{{\tt Z}_j}}$` and
\lstinline`neq$_{\rho_{{\tt Z}_j}}$` do not occur in $\Prog$.
Note that except for {\tt test$_{\mathsf{B}_i}$} the corresponding rules of these new predicates  (\ie rules for which these
predicates occur in the head) do not involve any other predicate and
consequently the meanings of these predicates do not depend on any other
predicate in $\prog{P}'$.
Therefore, every stable model of $\prog{P}'$ assigns the same
meaning to these predicates. This observation is
formalized in the following proposition.

\begin{proposition}
If $\mathcal{I}$ is a three-valued stable model of $\prog{P}'$ then
\begin{enumerate}
\item $\mathcal{I}({\tt empty}_{\rho_{\tt R}})(\bar{d}) = \mfalse$
\item $\mathcal{I}({\tt add}_{\rho_{\tt R}})(d_1,\bar{d_2},\bar{d_3}) = (d_1(\bar{d_3})) \bigvee (\bar{d_2} = \bar{d_3})$
\item $\mathcal{I}({\tt eq}_{\rho_{{\tt Z}_j}})(d_1,d_2) =
            \left\{\begin{array}{ll}
                \mtrue &  d_1=d_2 \\
                \mfalse & \mathit{otherwise}
            \end{array}\right.$
\item $\mathcal{I}({\tt neq}_{\rho_{{\tt Z}_j}})(d_1,d_2) =
            \left\{\begin{array}{ll}
                \mfalse &  d_1=d_2 \\
                \mtrue & \mathit{otherwise}
            \end{array}\right.$
\end{enumerate}
\end{proposition}


We will use the following slightly modified definition of \lstinline`test$_{\mathsf{B}_i}$`,
that contains a seemingly redundant rule which, however, helps 
simplify the proofs.
\begin{lstlisting}
test$_{\mathsf{B}_i}$($\bar{\tt X}$,R):-$\mathsf{B}_i[{\tt R}]$.
test$_{\mathsf{B}_i}$($\bar{\tt X}$,R):-test$_{\mathsf{B}_i}$($\bar{\tt X}$,R).
test$_{\mathsf{B}_i}$($\bar{\tt X}$,R):-test$_{\mathsf{B}_i}$($\bar{\tt X}$,add$_{\rho_{\tt R}}$(R,$\bar{\tt Z}$)).
\end{lstlisting}
It is easy to verify that in any stable model of $\prog{P}'$
the above definition has the same meaning  with the definition used in Section~\ref{transformation}.
%
\begin{lemma}\label{apspecial}
Let $(I,J) \in H^c_\mathsf{P'}$ such that $\tau^{-1}(I,J)$ assigns
the aforementioned meanings to the predicates ${\tt empty}_{\rho_{\tt R}}$ and ${\tt add}_{\rho_{\tt R}}$.
Then, the following hold:
\begin{itemize}
\item $A_{\prog{P}'}(I,J)_1({\tt test}_{\mathsf{B}_i})(\bar{x},r)
  = \bigvee\nolimits_{\leq_\bool} \{ I({\tt test}_{\mathsf{B}_i})(\bar{x},r') \mid r\leq r', {|r'|\leq |r|+1} \} \vee b_1$
\item $A_{\prog{P}'}(I,J)_2({\tt test}_{\mathsf{B}_i})(\bar{x},r)
  = \bigvee\nolimits_{\leq_\bool} \{ J({\tt test}_{\mathsf{B}_i})(\bar{x},r') \mid r\leq r', {|r'|\leq |r|+1} \} \vee b_2$
\end{itemize}
where $(b_1, b_2) =\bigvee\nolimits_{\leq_\bool}\{\mwrc{\mathsf{B}_i}{I,J}{s[\bar{\tt X}/\bar{x}, {\tt R}/r]} \mid  s \in S_{\prog{P}'} \}$
and $|r| = |\{ \bar{d} \in \lsem \rho_{\bar{\tt X}} \rsem \mid r(\bar{d}) = \mtrue \}|$.
\end{lemma}

The proof of Lemma~\ref{apspecial}, which is relatively quite straightforward 
but lengthy, is omitted. The proof uses the modified definition of \lstinline`test$_{\mathsf{B}_i}$` 
described above.

\begin{lemma}\label{IJSEQ}
Let $(I,J) \in H^c_\mathsf{P'}$ be a fixpoint of $A_{\prog{P}'}$ such that
$\tau^{-1}(I,J)$ assigns the aforementioned meanings to the predicates
${\tt empty}_{\rho_{\tt R}}$ and ${\tt add}_{\rho_{\tt R}}$.
If $I^* = \lfp A_{\prog{P}'}(\cdot,J)_1$ and $J^* =\lfp A_{\prog{P}'}(I,\cdot)_2$ then
for every $\bar{x} \in \lsem \rho_{\bar{\tt X}}\rsem$ and
$r \in \lsem \rho_{\tt R} \rsem$,
\begin{enumerate}
\item
\(
I^*({\tt test}_{\mathsf{B}_i})({\bar{x}},r) =
  \bigvee\nolimits_{\leq_\bool} \{ [\mwrc{\mathsf{B}_i}{I^*,J}{s[\bar{\tt X}/\bar{x}, {\tt R}/r']} ]_1 \mid {r \leq r' }, s \in S_{\prog{P}'} \}
\)

\item
\(
J^*({\tt test}_{\mathsf{B}_i})({\bar{x}},r) =
\bigvee\nolimits_{\leq_\bool} \{ [\mwrc{\mathsf{B}_i}{I,J^*}{s[\bar{\tt X}/\bar{x}, {\tt R}/r']} ]_2
        \vee I({\tt test}_{\mathsf{B}_i})(\bar{x},r') \mid {r \leq r' }, s \in S_{\prog{P}'} \}
\)
\end{enumerate}
\end{lemma}

\begin{proof}
We will prove the first part of the lemma in two steps. First, we will show that
$I^*({\tt test}_{\mathsf{B}_i})({\bar{x}},r) \geq
  \bigvee\nolimits_{\leq_\bool}\{ [\mwrc{\mathsf{B}_i}{I^*,J}{s[\bar{\tt X}/\bar{x}, {\tt R}/r']} ]_1 \mid {r \leq r' }, s \in S_{\prog{P}'} \}$
and then $I^*({\tt test}_{\mathsf{B}_i})({\bar{x}},r) \leq
\bigvee\nolimits_{\leq_\bool} \{ [\mwrc{\mathsf{B}_i}{I^*,J}{s[\bar{\tt X}/\bar{x}, {\tt R}/r']} ]_1 \mid {r \leq r' }, s \in S_{\prog{P}'} \}$
which establishes the equality.

For the first step, let $\bar{x} \in \lsem \rho_{\bar{\tt X}} \rsem$ and
$r \in \lsem \rho_{\tt R} \rsem$ such that
\(
  \bigvee\nolimits_{\leq_\bool}\{ [\mwrc{\mathsf{B}_i}{I^*,J}{s[\bar{\tt X}/\bar{x}, {\tt R}/r']}]_1 \mid {r \leq r' }, s \in S_{\prog{P}'} \} = \mtrue
\).
Therefore, there must be $s \in S_{\prog{P}'}$ and $t \in \lsem \rho_{\tt R} \rsem$
such that $t \geq r$ and $[\mwrc{\mathsf{B}_i}{I^*,J}{s[\bar{\tt X}/\bar{x}, {\tt R}/t]}]_1 = \mtrue$.
Notice that auxiliary predicates  ${\tt empty}_{\rho_{\tt R}}$ and ${\tt add}_{\rho_{\tt R}}$ in $I^*,J^*$ will also
have the aforementioned meaning therefore by the definition of $A_{\prog{P}'}$ with Lemma~\ref{apspecial} 
and the fact that $I^*=A_{\prog{P}'}(I^*,J)_1$ it follows that
\begin{align*}
I^*({\tt test}_{\mathsf{B}_i})(\bar{x},r)
 & = \bigvee\nolimits_{\leq_\bool}\{ I^*({\tt test}_{\mathsf{B}_i})(\bar{x},r') \mid r\leq r', {|r'|\leq |r|+1} \}
 \vee b_1 \\
& \geq
\bigvee\nolimits_{\leq_\bool}\{ I^*({\tt test}_{\mathsf{B}_i})(\bar{x},r') \mid r\leq r', {|r'|\leq |r|+1} \}
\end{align*}
where $b_1 = \bigvee\nolimits_{\leq_\bool} \{ [\mwrc{\mathsf{B}_i}{I^*,J}{s[\bar{\tt X}/\bar{x}, {\tt R}/r]}]_1 \mid s \in S_{\prog{P}'} \}$.
By unfolding the inequality $|t|-|r|$ times
\begin{align*}
I^*({\tt test}_{\mathsf{B}_i})(\bar{x},r)
  &\geq \bigvee\nolimits_{\leq_\bool}\{ I^*({\tt test}_{\mathsf{B}_i})(\bar{x},r') \mid r\leq r', {|r'|\leq |t|} \} \\
  &\geq I^*({\tt test}_{\mathsf{B}_i})(\bar{x},t)\\
  &\geq [\mwrc{\mathsf{B}_i}{I^*,J}{s[\bar{\tt X}/\bar{x}, {\tt R}/t]}]_1 =\mtrue
\end{align*}
Therefore,
$I^*({\tt test}_{\mathsf{B}_i})(\bar{x},r) \geq
\bigvee \{ [\mwrc{\mathsf{B}_i}{I^*,J}{s[\bar{\tt X}/\bar{x}, {\tt R}/r']} ]_1 \mid {r \leq r' }, s \in S_{\prog{P}'}\}$.

Now we show the second step of the first statement.
Let $I'$ such that
$I'({\tt test}_{\mathsf{B}_i})(\bar{x},r) =
  \bigvee\nolimits_{\leq_\bool}\{ [\mwrc{\mathsf{B}_i}{I^*,J}{s[\bar{\tt X}/\bar{x}, {\tt R}/r']}]_1
            \mid {r \leq r' }, s \in S_{\prog{P}'} \}$
and $I'({\tt q})=I^*({\tt q})$ for every other predicate ${\tt q}$ occurring in $\prog{P}'$.
We will show that $I'$ is a prefixpoint of $A_{\prog{P}'}(\cdot,J)_1$; 
since $I^*$ is the least prefixpoint of $A_{\prog{P}'}(\cdot,J)_1$ it will 
follow that $I^*({\tt test}_{\mathsf{B}_i})(\bar{x},r)) \leq I'({\tt test}_{\mathsf{B}_i})(\bar{x},r)$
that establishes the desired claim.
For any predicate ${\tt q}$ in $\Prog$ different from ${\tt p}$ and ${\tt test}_{\mathsf{B}_i}$,
we have that
\(
A_{\prog{P}'}(I',J)_1({\tt q})= A_{\prog{P}'}(I^*,J)_1({\tt q})=I^*({\tt q})=I'({\tt q})
\).
For ${\tt test}_{\mathsf{B}_i}$ we have that
\[
A_{\prog{P}'}(I',J)_1({\tt test}_{\mathsf{B}_i})(\bar{x},r) =
 \bigvee\nolimits_{\leq_\bool}\{ I'({\tt test}_{\mathsf{B}_i})(\bar{x},r') \mid {r\leq r'}, {|r'|\leq |r|+1} \}
\vee b_1
\]
where $b_1 = \bigvee\nolimits_{\leq_\bool} \{ [\mwrc{\mathsf{B}_i}{I',J}{s[\bar{\tt X}/\bar{x}, {\tt R}/r]}]_1 \mid s \in S_{\prog{P}'} \}
=\bigvee\nolimits_{\leq_\bool} \{ [\mwrc{\mathsf{B}_i}{I^*,J}{s[\bar{\tt X}/\bar{x}, {\tt R}/r]}]_1 \mid s \in S_{\prog{P}'} \}$
(since in the body of $\mathsf{B}_i$ predicate ${\tt test}_{\mathsf{B}_i}$ does not appear).
But we can rewrite the above by replacing $I'({\tt test}_{\mathsf{B}_i})(\bar{x},r')$ as:
\begin{align*}
A_{\prog{P}'}(I',J)_1({\tt test}_{\mathsf{B}_i})(\bar{x},r)
  & = \bigvee\nolimits_{\leq_\bool} \{
      [\mwrc{\mathsf{B}_i}{I^*,J}{s[\bar{\tt X}/\bar{x}, {\tt R}/r']}]_1 \mid r \leq r', s \in S_{\prog{P}'} \}
      \vee b_1
\end{align*}
Since the term $b_1$ is also absorbed this is equivalent to the following:
\begin{align*}
A_{\prog{P}'}(I',J)_1({\tt test}_{\mathsf{B}_i})(\bar{x},r)
  & = \bigvee\nolimits_{\leq_\bool}\{ [\mwrc{\mathsf{B}_i}{I^*,J}{s[\bar{\tt X}/\bar{x}, {\tt R}/r']}]_1 \mid {r \leq r' }, s \in S_{\prog{P}'} \}
   = I'({\tt test}_{\mathsf{B}_i})(\bar{x},r)
\end{align*}
Finally, for ${\tt p}$ let \lstinline`p($\bar{\tt X}$):-$\mathsf{B}_j$` be the rules in $\Prog$ and let $i$-th be the rule
that is replaced in $\prog{P}'$. Then, it follows:
\begin{align*}
  A_{\prog{P}'}(I',J)_1({\tt p})(\bar{x})
        &=\bigvee\nolimits_{\leq_\bool} \{\mwrc{\mathsf{B}_j}{I',J}{s[\bar{\tt X}/\bar{x}]}]_1 \mid j \neq i, s \in S_{\prog{P}'} \}
        \vee
        I'({\tt test}_{\mathsf{B}_i})(\bar{x},\bot_{\tt R}) \\
        &\leq \bigvee\nolimits_{\leq_\bool} \{ [\mwrc{\mathsf{B}_j}{I^*,J}{s[\bar{\tt X}/\bar{x}]}]_1 \mid j \neq i, s \in S_{\prog{P}'} \}
        \vee
        I^*({\tt test}_{\mathsf{B}_i})(\bar{x},\bot_{\tt R})\\
        &\leq A_{\prog{P}'}(I^*,J)_1({\tt p})(\bar{x})
\end{align*}
Since $I^*$ is a fixpoint of $A_{\prog{P}'}(\cdot,J)_1$ it follows that
$A_{\prog{P}'}(I^*,J)_1({\tt p})(\bar{x}) = I^*({\tt p})(\bar{x})$. Moreover,
by definition of $I'$, $I'({\tt p})(\bar{x}) = I^*({\tt p})(\bar{x})$ and therefore
$A_{\prog{P}'}(I',J)_1({\tt p})(\bar{x}) \leq I'({\tt p})(\bar{x})$.


Similarly to the proof for the first statement we prove the second one 
in two steps.
Let $\bar{x},r$ such that
$ \bigvee\nolimits_{\leq_\bool} \{
  I({\tt test}_{\mathsf{B}_i})(\bar{x},r')
  \vee
  [\mwrc{\mathsf{B}_i}{I,J^*}{s[\bar{\tt X}/\bar{x}, {\tt R}/r']}]_2
  \mid {r \leq r' , s \in S_{\prog{P}'}}\}=\mtrue$.
Then, there exists $s \in S_{\prog{P}'}$ and $t\geq r$ such that
$I({\tt test}_{\mathsf{B}_i})(\bar{x},t) \vee
[\mwrc{\mathsf{B}_i}{I,J^*}{s[\bar{\tt X}/\bar{x}, {\tt R}/t]}]_2 = \mtrue$.
By the definition of the operator $A_{\prog{P}'}$ using Lemma~\ref{apspecial} and the fact that $J^* = A_{\prog{P}'}(I,J^*)_2$, it follows that:
\begin{align*}
J^*({\tt test}_{\mathsf{B}_i})(\bar{x},r) & =
\bigvee\nolimits_{\leq_\bool} \{ J^*({\tt test}_{\mathsf{B}_i})(\bar{x},r') \mid {r\leq r'}, {|r'|\leq |r|+1}, s \in S_ {\prog{P}'} \}
\vee b_2 \\
& \geq
\bigvee\nolimits_{\leq_\bool} \{ J^*({\tt test}_{\mathsf{B}_i})(\bar{x},r') \mid {r\leq r'}, {|r'|\leq |r|+1}, s \in S_ {\prog{P}'} \}
\end{align*}
where $b_2 = \bigvee\nolimits_{\leq_\bool} \{ [\mwrc{\mathsf{B}_i}{I,J^*}{s[\bar{\tt X}/\bar{x}, {\tt R}/r]}]_2 \mid s \in S_{\prog{P}'} \}$.
By unfolding the inequality $|t|-|r|$ times
\begin{align*}
J^*({\tt test}_{\mathsf{B}_i})(\bar{x},r)
    & \geq \bigvee\nolimits_{\leq_\bool} \{ J^*({\tt test}_{\mathsf{B}_i})(\bar{x},r') \mid {r\leq r'}, {|r'|\leq |t|}, s \in S_ {\prog{P}'} \}    \\
    &\geq J^*({\tt test}_{\mathsf{B}_i})(\bar{x},t)
    \geq [\mwrc{\mathsf{B}_i}{I,J^*}{s[\bar{\tt X}/ \bar{x}, {\tt R}/ t]}]_2
\end{align*}
Furthermore, since $J^* \geq I$ we have that $J^*({\tt test}_{\mathsf{B}_i})(\bar{x},r) 
\geq J^*({\tt test}_{\mathsf{B}_i})(\bar{x},t)\geq I({\tt test}_{\mathsf{B}_i})(\bar{x},t)$ and so
\begin{align*}
J^*({\tt test}_{\mathsf{B}_i})(\bar{x},r)
    \geq I({\tt test}_{\mathsf{B}_i})(\bar{x},t)
    \vee
    [\mwrc{\mathsf{B}_i}{I,J^*}{s[\bar{\tt X}/ \bar{x}, {\tt R}/ t]}]_2 =\mtrue
\end{align*}


Let $J'$ such that $J'({\tt q})=J({\tt q})$
for any ${\tt q}$ in $\Prog$ and
$J'({\tt test}_{\mathsf{B}_i})(\bar{x}, r) =
      \bigvee \{ [\mwrc{\mathsf{B}_i}{I,J^*}{s[\bar{\tt X}/\bar{x}, {\tt R}/r']} ]_2
                  \vee I({\tt test}_{\mathsf{B}})(\bar{x},r')
                  \mid {r \leq r' }, s \in S_{\prog{P}'}
              \}$.
We will show that $J'$ is a prefixpoint of $A_{\prog{P}'}(I, \cdot)$ and therefore
$J^* \leq J'$ since $J^*$ is the least prefixpoint of $A_{\prog{P}'}(I, \cdot)$.
Clearly, $J' \geq I$ by its definition.
For any predicate ${\tt q}$ in $\Prog$ different from ${\tt p}$ and ${\tt test_{\mathsf{B}_i}}$,
we have $A_{P'}(I,J')_2({\tt q})= A_{P'}(I,J^*)_2({\tt q})=J^*({\tt q})=J'({\tt q})$.
For the predicate ${\tt test}_{\mathsf{B}_i}$ similarly as the case of $I$ we have that:
\begin{align*}
A_{\prog{P}'}(I,J')_2({\tt test}_{\mathsf{B}_i})(\bar{x},r)
  & = \bigvee\nolimits_{\leq_\bool} \{ J'({\tt test}_{\mathsf{B}_i})(\bar{x},r') \mid {r\leq r'}, {|r'|\leq |r|+1}, s \in S_ {\prog{P}'} \} \vee b_2\\
  & = \bigvee\nolimits_{\leq_\bool} \{ I({\tt test}_{\mathsf{B}_i})(\bar{x},r') \vee \mwrc{\mathsf{B}_i}{I,J^*}{s[\bar{\tt X}/\bar{x}, \bar{R}/r']}\mid {r\leq r'}, s \in S_ {\prog{P}'} \} \vee b_2\\
  & = \bigvee\nolimits_{\leq_\bool} \{ I({\tt test}_{\mathsf{B}_i})(\bar{x},r') \vee \mwrc{\mathsf{B}_i}{I,J^*}{s[\bar{\tt X}/\bar{x}, \bar{R}/r']}\mid {r\leq r'}, s \in S_ {\prog{P}'} \} \\
  & = J'({\tt test}_{\mathsf{B}_i})(\bar{x},r)
\end{align*}
Finally, for ${\tt p}$ let \lstinline`p($\bar{\tt X}$):-$\mathsf{B}_j$` be the rules in $\Prog$ and let $i$-th be the rule
that is replaced in $\prog{P}'$. Then, it follows:
\begin{align*}
A_{\prog{P}'}(I,J')_2({\tt p})(\bar{x})
      &=\bigvee\nolimits_{\leq_\bool} \{\mwrc{\mathsf{B}_j}{I,J'}{s[\bar{\tt X}/\bar{x}]}]_2 \mid j \neq i, s \in S_{\prog{P}'} \}
                  \vee
                  J'({\tt test}_{\mathsf{B}_i})(\bar{x},\bot_{\tt R}) \\
      &\leq \bigvee\nolimits_{\leq_\bool} \{ [\mwrc{\mathsf{B}_j}{I,J^*}{s[\bar{\tt X}/\bar{x}]}]_2 \mid j \neq i, s \in S_{\prog{P}'} \}
      \vee
      J^*({\tt test}_{\mathsf{B}_i})(\bar{x},\bot_{\tt R})\\
      &\leq A_{\prog{P}'}(I,J^*)_2({\tt p})(\bar{x})
\end{align*}
Since $J^*$ is a fixpoint of $A_{\prog{P}'}(I, \cdot)_2$ it follows that
$A_{\prog{P}'}(I,J^*)_2({\tt p})(\bar{x}) = J^*({\tt p})(\bar{x})$. Moreover,
by definition of $J'$, $J'({\tt p})(\bar{x}) = J^*({\tt p})(\bar{x})$ and therefore
$A_{\prog{P}'}(I,J')_2({\tt p})(\bar{x}) \leq J'({\tt p})(\bar{x})$.
\end{proof}

\begin{corollary}\label{stable-meaning-of-testb}
Let $(I,J) \in H^c_\mathsf{P'}$
such that $(I,J)$ is a stable fixpoint of $A_{\prog{P}'}$. Then,
for every $\bar{x} \in \mo{\rho_{\bar{\tt X}}}$ and $r \in \mo{\rho_{\tt R}}$,
\begin{itemize}
\item  $I({\tt test}_{\mathsf{B}_i})(\bar{x},r) =
\bigvee\nolimits_{\leq_\bool} \{ [\mwrc{\mathsf{B}_i}{I,J}{s[\bar{\tt X}/\bar{x}, {\tt R}/r']}]_1
            \mid {r \leq r' }, s \in S_{\prog{P}'} \}$.
\item  $J({\tt test}_{\mathsf{B}_i})(\bar{x},r) =
\bigvee\nolimits_{\leq_\bool} \{ [\mwrc{\mathsf{B}_i}{I,J}{s[\bar{\tt X}/\bar{x}, {\tt R}/r']}]_2
          \mid {r \leq r' }, s \in S_{\prog{P}'} \}$.
\end{itemize}
\end{corollary}

In other words, every three-valued stable model of $\prog{P}'$ assigns the aforementioned
meaning to predicate ${\tt test}_{\mathsf{B}_i}$. In the following we
define a mapping $\theta$ that extends any interpretation of $\Prog$
to an interpretation of $\prog{P}'$ that assigns specific meanings to
the predicates ${\tt empty}_{\rho_{\tt R}}$, ${\tt add}_{\rho_{\tt R}}$,
${\tt eq}_{\rho_{{\tt Z}_j}}$, ${\tt neq}_{\rho_{{\tt Z}_j}}$
and ${\tt test}_{\mathsf{B}_i}$.
More precisely, if $\mathcal{I}$ is an interpretation of $\Prog$ then
$\theta(\mathcal{I})$ is defined as follows:
\begin{itemize}
\item $\theta(\mathcal{I})({\tt q})  = \mathcal{I}({\tt q}) \text{\ if ${\tt q} \text{ occurs in } \Prog$}$
\item $\theta(\mathcal{I})({\tt empty}_{\rho_{\tt R}}) = \bot_{\rho_{\tt R}}$
\item $\theta({\mathcal I})({\tt add}_{\rho_{\tt R}})(d_1,\bar{d_2},\bar{d_3}) = (d_1(\bar{d_3})) \bigvee (\bar{d_2} = \bar{d_3})$
\item \( \theta({\mathcal I})({\tt eq}_{\rho_{{\tt Z}_j}})({d_1},{d_2})  =
            \left\{\begin{array}{ll}
                \mtrue &  d_1= d_2 \\
                \mfalse & \mathit{otherwise}
            \end{array}\right. \)
\item \( \theta({\mathcal I})({\tt neq}_{\rho_{{\tt Z}_j}})({d_1},{d_2})  =
            \left\{\begin{array}{ll}
                \mfalse & d_1 = d_2 \\
                \mtrue  & \mathit{otherwise}
            \end{array}\right. \)
\item $\theta(\mathcal{I})({\tt test}_{\mathsf{B}_i})({\bar{x}},r)  =
\bigvee\nolimits_{\leq_\bool} \{ \mos{\mathsf{B}_i}_{s[\bar{\tt X}/\bar{x}, {\tt R}/r']}(\mathcal{I}) \mid r \leq r', s \in S_{\prog{P}'}  \}$
\end{itemize}
Conversely, the mapping $\theta^{-1}$ defined as
$\theta^{-1}(\mathcal{I'})({\tt q})=\mathcal{I'}({\tt q})$
for all predicates \lstinline`q` occurring in $\Prog$, maps an interpretation $\mathcal{I}'$
of $\prog{P}'$ to an interpretation $\theta^{-1}(\mathcal{I}')$ of $\Prog$.

The following lemma establishes that $\theta$ and $\theta^{-1}$ preserve both the $\leq$ and $\preceq$ orderings
and is easy to verify, so the formal proof is omitted.
\begin{lemma}\label{theta-order-preserving}
Let $\mathcal{I}_1, \mathcal{I}_2$ be interpretations of $\Prog$ and
$\mathcal{I}'_1, \mathcal{I}'_2$ be interpretations of $\prog{P}'$.
\begin{enumerate}
  \item If $\mathcal{I}_1 \leq \mathcal{I}_2$ then $\theta(\mathcal{I}_1) \leq \theta(\mathcal{I}_2)$.
  \item If $\mathcal{I}'_1 \leq \mathcal{I}'_2$ then $\theta^{-1}(\mathcal{I}'_1) \leq \theta^{-1}(\mathcal{I}'_2)$.
  \item If $\mathcal{I}_1 \preceq \mathcal{I}_2$ then $\theta(\mathcal{I}_1) \preceq \theta(\mathcal{I}_2)$.
  \item If $\mathcal{I}'_1 \preceq \mathcal{I}'_2$ then $\theta^{-1}(\mathcal{I}'_1) \preceq \theta^{-1}(\mathcal{I}'_2)$.
\end{enumerate}
\end{lemma}

The following Lemma is also easy to verify and builds on top of Corollary~\ref{stable-meaning-of-testb} and the definition 
of the mappings $\theta$ and $\theta^{-1}$.
\begin{lemma}\label{theta-properties}\label{fixpoints}
\begin{enumerate}
\item If $\mathcal{I}$ is a stable model of $\Prog$ then $\mathcal{I} = \theta^{-1}(\theta(\mathcal{I}))$.
\item If $\mathcal{I}'$ is a stable model of $\prog{P}'$ then $\mathcal{I}' = \theta(\theta^{-1}(\mathcal{I}'))$.
\item If $\tau(\mathcal{I})$ is a stable fixpoint of $A_\Prog$ then $\tau(\theta(\mathcal{I}))$ is a fixpoint of $A_{\prog{P}'}$.
\item If  $\tau(\mathcal{I}')$ is a stable fixpoint of $A_{\prog{P}'}$ then $\tau(\theta^{-1}(\mathcal{I}'))$ is a fixpoint of $A_\Prog$.
\item If $\mathcal{I}$ and $\mathcal{I}'$ are two-valued then $\theta(\mathcal{I})$ and $\theta^{-1}(\mathcal{I}')$ are also two-valued.
\end{enumerate}
\end{lemma}




\begin{lemma}\label{rightdirection}
If $\mathcal{I}$ is a \emph{three-valued} stable model of $\Prog$
then $\theta(\mathcal{I})$ is a \emph{three-valued} stable model of $\prog{P}'$.
\end{lemma}

\begin{proof}
Let $\mathcal{I} = \tau^{-1}(I,J)$ and $\theta(\mathcal{I}) = \tau^{-1}(I',J')$.
It is easy to see that $\theta^{-1}(I') = I$ and $\theta^{-1}(J') = J$\footnote{Note that $I',J'$ are two valued interpretations therefore also ``three-valued''
and belong to the domain of $\theta^{-1}$. In addition, since $I',J'$ are two-valued so is $\theta^{-1}(I')$ and $\theta^{-1}(J')$.}.
By Lemma~\ref{fixpoints}, $(I',J')$ is a fixpoint of $A_{\prog{P}'}$
and therefore, $A_{\prog{P}'}(\cdot,J')_1$ and $A_{\prog{P}'}(I',\cdot)_2$
are well-defined monotone operators in the corresponding intervals.
By Definition~\ref{def:AFTsemantics},
it suffices to show that $I' = \lfp A_{\prog{P}'}(\cdot,J')_1$ and
$J' = \lfp A_{\prog{P}'}(I', \cdot)_2$. We will prove
each of these by contradiction.

Let $I^* = \lfp A_{\prog{P}'}(\cdot,J')_1$
and assume for the purpose of contradiction that
$I^* < I'$, \ie there is a predicate ${\tt q'}$ in $\prog{P}'$ such that
$I^*({\tt q'}) < I'({\tt q'})$. Clearly ${\tt q'}$ cannot be any of the auxiliary predicates 
${\tt empty}_{\rho_{\tt R}}$ and ${\tt add}_{\rho_{\tt R}}$.
We proceed with case analysis on ${\tt q'}$.
Assume that the only predicate ${\tt q'}$ that differs is ${\tt test}_{\mathsf{B}_i}$.
This is a contradiction because by Lemma~\ref{IJSEQ} we have that
\(I^*({\tt test}_{\mathsf{B}_i})(\bar{x},r)
    = \bigvee\nolimits_{\leq_\bool} \{ [\mwrc{\mathsf{B}_i}{I^*,J'}{s[\bar{\tt X}/\bar{x}, {\tt R}/r']}]_1 \mid
                {r \leq r', s \in S_{\prog{P}'} }\}
\).
Since $\mathsf{B}_i$ does not contain ${\tt test}_{{\mathsf{B}}_i}$, it follows that
for any $s \in S_{\prog{P}'}$, $\mwrc{\mathsf{B}_i}{I^*,J'}{s} = \mwrc{\mathsf{B}_i}{I',J'}{s}= \mwrc{\mathsf{B}_i}{I,J}{s}$ and therefore
$I^*({\tt test}_{{\mathsf{B}}_i}) = \bigvee\{ [\mwrc{\mathsf{B}_i}{I,J}{s[\bar{\tt X}/\bar{x}, {\tt R}/r']}]_1 \mid
                {r \leq r', s \in S_{\prog{P}'} }\}=I'({\tt test}_{{\mathsf{B}}_i})$.
We conclude that there must be at least one predicate ${\tt q'}$ that is common in both $\Prog$ and $\prog{P}'$ such that
$I^*({\tt q'}) < I'({\tt q'})$. If this is true it follows that $\theta^{-1}(I^*) < \theta^{-1}(I') = I$.
We will show that $\theta^{-1}(I^*)$ is a fixpoint of $A_{\Prog}(\cdot,J)_1$
which is a contradiction since $I$ is the least fixpoint.
Indeed, for any predicate ${\tt q}$ occurring in $\Prog$ different from ${\tt p}$ (which was the predicate affected by the transformation),
\( A_{\Prog}(\theta^{-1}(I^*),J)_1({\tt q})
  = A_{\prog{P}'}(I^*,J')_1({\tt q})  \) since the rules of {\tt q} are the same
in both $\Prog$ and $\prog{P}'$ and the predicates appearing in the bodies of these rules
have the same meaning in both $I^*$ and $\theta^{-1}(I^*)$ and the same meaning in both $J'$ and $J$.
But $I^*$ is a fixpoint of $A_{\prog{P}'}(\cdot,J')_1$ and therefore
$A_{\prog{P}'}(I^*,J')_1({\tt q}) = I^*({\tt q})$. But since
$I^*({\tt q}) = \theta^{-1}(I^*)({\tt q})$ it follows that
$\theta^{-1}(I^*)({\tt q}) = A_{\prog{P}}(\theta^{-1}(I^*),J)_1({\tt q})$.
Finally, for {\tt p}, it holds, by definition, that
$A_{\Prog}(\theta^{-1}(I^*),J)_1({\tt p})(\bar{x}) =
 \bigvee\{ [\mwrc{\mathsf{B}_j}{\theta^{-1}(I^*),J}{s[\bar{\tt X}/\bar{x}]}]_1 \mid ({\tt p(\bar{X})} \leftarrow \mathsf{B}_j) \in \Prog, s \in S_\Prog \}$.
Notice that each $\mathsf{B}_j$ contains only predicates in $\Prog$ and therefore,
for any $s \in S_{\Prog}$, $\mwrc{\mathsf{B}_j}{\theta^{-1}(I^*),J}{s} = \mwrc{\mathsf{B}_j}{I^*,J'}{s}$.
Moreover, by Lemma~\ref{IJSEQ},
$I^*({\tt test}_{\mathsf{B}_i})(\bar{x},\bot_{\rho_{\tt R}})
  = \bigvee\nolimits_{\leq_\bool}\{[\mwrc{\mathsf{B}_i}{I^*,J'}{s[\bar{\tt X}/\bar{x}, {\tt R}/r']}]_1 \mid \bot_{\rho_{\tt R}} \leq r', s \in S_{\prog{P}'}\}
  = \bigvee\nolimits_{\leq_\bool}\{ [\mwrc{\mathsf{B}_i}{\theta^{-1}(I^*),J}{s[\bar{\tt X}/\bar{x}]}]_1 \mid s \in S_\Prog \}$.
If we distinguish the $i$-th rule from all the others we can rewrite the above as:
\begin{align*}
A_{\Prog}(\theta^{-1}(I^*),J)_1({\tt p})(\bar{x})
  & = \bigvee\nolimits_{\leq_\bool} \{ [\mwrc{\mathsf{B}_j}{\theta^{-1}(I^*),J}{s[\bar{\tt X}/\bar{x}]}]_1 \mid ({\tt p(\bar{X})} \leftarrow \mathsf{B}_j) \in \Prog, s \in S_\Prog \} \\
  & = I^*({\tt test}_{\mathsf{B}_i})(\bar{x},\bot_{\rho_{\tt R}}) \vee 
  \bigvee\nolimits_{\leq_\bool}  \{ [\mwrc{\mathsf{B}_j}{I^*,J'}{s[\bar{\tt X}/\bar{x}]}]_1 \mid {j \neq i}, s \in S_{\prog{P}'} \}\\
  & = A_{\prog{P}'}(I^*,J')_1({\tt p})(\bar{x}) 
\end{align*}
Recall that $I^*$ is a fixpoint of $A_{\prog{P}'}(\cdot,J')_1$ and also $I^*({\tt p})(\bar{x}) = \theta^{-1}(I^*)({\tt p})(\bar{x})$.
Therefore, $A_{\Prog}(\theta^{-1}(I^*),J)_1({\tt p}) = \theta^{-1}(I^*)({\tt p})$. But then we conclude that
$\theta^{-1}(I^*)$ must be a  fixpoint which is a contradiction since we assumed that $\theta^{-1}(I^*) < I$ and $I$ is the least fixpoint.

We proceed to prove that $J'$ is the least fixpoint of $A_{\prog{P}'}(I',\cdot)_2$. The proof
is similar to the proof for $I'$.
Let $J^* =\lfp A_{\prog{P}'}(I',\cdot)_2$ and assume that $J^* < J'$, \ie there is a
predicate ${\tt q'}$ in $\Prog$ such that $J^*({\tt q'}) < J'({\tt q'})$.
As the previous case, assume that the only ${\tt q'}$ that differs is ${\tt test}_{\mathsf{B}}$. This is a contradiction as we will soon see.
By Lemma~\ref{IJSEQ} we have that
\begin{align*}
J^*({\tt test}_{\mathsf{B}_i})({\bar{x}},r)
  &= \bigvee\nolimits_{\leq_\bool} \{ I'({\tt test}_{\mathsf{B}_i})(\bar{x},r')
      \vee [\mwrc{\mathsf{B}_i}{I',J^*}{s[\bar{\tt X}/\bar{x}, {\tt R}/r]} ]_2  \mid {r \leq r' }, s \in S_{\prog{P}'} \}\\
\end{align*}
with 
\begin{align*}
I'({\tt test}_{\mathsf{B}_i})({\bar{x}},r) & =
  \bigvee\nolimits_{\leq_\bool} \{ [\mwrc{\mathsf{B}_i}{I,J}{s(\bar{\tt X}/\bar{x}, {\tt R}/r')} ]_1 \mid {r \leq r'}, s \in S_{\prog{P}'} \}
\end{align*}
and therefore $J^*({\tt test}_{\mathsf{B}_i})({\bar{x}},r)$ can be rewritten as
\[
  J^*({\tt test}_{\mathsf{B}_i})({\bar{x}},r) =
  \bigvee\nolimits_{\leq_\bool} \{ [\mwrc{\mathsf{B}_i}{I,J}{s[\bar{\tt X}/\bar{x}, {\tt R}/r]} ]_1
      \vee [\mwrc{\mathsf{B}_i}{I',J^*}{s[\bar{\tt X}/\bar{x}, {\tt R}/r]} ]_2 \mid {r \leq r' }, s \in S_{\prog{P}'} \}
\]
Since $\mathsf{B}_i$ does not contain ${\tt test}_{\mathsf{B}_i}$, it follows that
for any $s \in S_{\prog{P}'}$, $\mwrc{\mathsf{B}_i}{I,J}{s} =\mwrc{\mathsf{B}_i}{I',J'}{s} = \mwrc{\mathsf{B}_i}{I',J^*}{s}$.
Moreover, $[\mwrc{\mathsf{B}_i}{I',J^*}{s}]_1 \leq [\mwrc{\mathsf{B}_i}{I',J^*}{s}]_2$ and therefore
$[\mwrc{\mathsf{B}_i}{I,J}{s}]_1 \leq [\mwrc{\mathsf{B}_i}{I',J^*}{s}]_2$. We can rewrite the above
as
\begin{align*}
J^*({\tt test}_{\mathsf{B}_i})({\bar{x}},r)
& = \bigvee\nolimits_{\leq_\bool}\{ [\mwrc{\mathsf{B}_i}{I',J^*}{s[\bar{\tt X}/\bar{x}, {\tt R}/r]} ]_2 \mid {r \leq r' }, s \in S_{\prog{P}'} \} \\
& = \bigvee\nolimits_{\leq_\bool}\{ [\mwrc{\mathsf{B}_i}{I,J}{s[\bar{\tt X}/\bar{x}, {\tt R}/r]} ]_2 \mid {r \leq r' }, s \in S_{\prog{P}'} \}
 = J'({\tt test}_{\mathsf{B}_i})({\bar{x}},r)
\end{align*}
This leads to a contradiction simply because the assumption was that $J^*$ and $J'$ differ only on ${\tt test}_{\mathsf{B}_i}$.
It must be the case that there exists at least one predicate ${\tt q'}$ common in both $\Prog$ and $\prog{P}'$ such that
$J^*({\tt q'}) < J'({\tt q'})$. But then it is $\theta^{-1}(J^*)< J$. We will show that
$\theta^{-1}(J^*)$ is a prefixpoint of $A_\Prog(I, \cdot)_2$ which is a contradiction since $J$ is 
the least fixpoint and therefore the least prefixpoint of $A_\Prog(I, \cdot)_2$.
Assume ${\tt q}$ is any predicate different than ${\tt p}$.
Then, $A_\Prog(I,\theta^{-1}(J^*))_2({\tt q}) = A_{\prog{P}'}(I',J^*)_2({\tt q})$ since the rules of {\tt q}
occurring in $\Prog$ are the same in both $\Prog$ and $\prog{P}'$ and the predicates appearing in the bodies of these rules
have the same meaning in both $I'$ and $I$ and the same meaning in both $J^*$ and $\theta^{-1}(J^*)$.
But $J^*$ is a fixpoint of $A_{\prog{P}'}(I',\cdot)_2$,
so $A_\Prog(I',J^*)_2({\tt q}) = J^*({\tt q})$. But since $J^*({\tt q})=\theta^{-1}(J^*)({\tt q})$
it follows that $\theta^{-1}(J^*)({\tt q}) = A_\Prog(I,\theta^{-1}(J^*))_2({\tt q})$.
Finally, for {\tt p},
it holds, by definition, that
$A_{\Prog}(I, \theta^{-1}(J^*))_2({\tt p})(\bar{x}) 
= \bigvee\nolimits_{\leq_\bool}\{ [\mwrc{\mathsf{B}_j}{I, \theta^{-1}(J^*)}{s[\bar{\tt X}/\bar{x}]}]_2 
\mid ({\tt p(\bar{X})} \leftarrow \mathsf{B}_j) \in \Prog, s \in S_\Prog \}$.
Notice that each $\mathsf{B}_j$ contains only predicates in $\Prog$ and therefore,
for any $s \in S_{\Prog}$, $\mwrc{\mathsf{B}_j}{I, \theta^{-1}(J^*)}{s} = \mwrc{\mathsf{B}_j}{I', J^*}{s}$.
But, by Lemma~\ref{IJSEQ},
$J^*({\tt test}_{\mathsf{B}_i})(\bar{x},\bot_{\rho_{\tt R}})
  \geq \bigvee\nolimits_{\leq_\bool}\{ [\mwrc{\mathsf{B}_i}{I',J^*}{s[\bar{\tt X}/\bar{x},{\tt R}/r']}]_2 \mid \bot_{\rho_{\tt R}} \leq r', s \in S_\Prog\}
  = \bigvee\nolimits_{\leq_\bool}\{ \mwrc{\mathsf{B}_i}{I',J^*}{s[\bar{\tt X}/\bar{x}]} \mid s \in S_\Prog \}$.
If we distinguish the rules of ${\tt p}$ we can rewrite the above as:
\begin{align*}
A_{\Prog}(I,\theta^{-1}(J^*))_2({\tt p})(\bar{x})
  & =\bigvee\nolimits_{\leq_\bool}\{ [\mwrc{\mathsf{B}_j}{I,\theta^{-1}(J^*)}{s[\bar{\tt X}/\bar{x}]}]_2 \mid ({\tt p(\bar{X})} \leftarrow \mathsf{B}_j) \in \Prog, s \in S_\Prog \} \\
  & \leq \bigvee\nolimits_{\leq_\bool}\{ [\mwrc{\mathsf{B}_j}{I',J^*}{s[\bar{\tt X}/\bar{x}]}]_2 \mid {j \neq i}, s \in S_{\prog{P}'} \} 
  \vee J^*({\tt test}_{\mathsf{B}_i})(\bar{x},\bot_{\rho_{\tt R}})\\
  & \leq A_{\prog{P}'}(I',J^*)_2({\tt p})(\bar{x}) 
\end{align*}
But $J^*$ is a fixpoint of $A_{\prog{P}'}(I',\cdot)_2$ and $J^*({\tt p})(\bar{x}) = \theta^{-1}(J^*)({\tt p})(\bar{x})$.
Therefore, $A_{\Prog}(I',\theta^{-1}(J^*))_2(p) \leq \theta^{-1}(J^*)(p)$. We conclude that $\theta^{-1}(J^*)$ is a prefixpoint
of $A_{\Prog}(I',\cdot)_2$ which is a contradiction since $\theta^{-1}(J^*) < J$ and $J$ is the least prefixpoint.
\end{proof}

\enlargethispage{1\baselineskip}
\begin{lemma}\label{leftdirection}
If $\mathcal{I}'$ is a \emph{three-valued} stable model of $\prog{P}'$ then
$\theta^{-1}(\mathcal{I}')$ is a \emph{three-valued} stable model of $\Prog$.
\end{lemma}

\begin{proof}
Let $\tau(\mathcal{I}') = (I', J')$ and $\tau(\theta^{-1}(\mathcal{I}')) = (I,J)$.
Since $(I',J')$ is a stable fixpoint of $A_{\prog{P}'}$, it follows by Corrollary~\ref{stable-meaning-of-testb} that
\begin{itemize}
\item  $I'({\tt test}_{\mathsf{B}_i})(\bar{x},r) =
\bigvee\nolimits_{\leq_\bool} \{ [\mwrc{\mathsf{B}_i}{I',J'}{s[\bar{\tt X}/\bar{x}, {\tt R}/r']}]_1
            \mid {r \leq r' }, s \in S_{\prog{P}'} \}$
\item  $J'({\tt test}_{\mathsf{B}_i})(\bar{x},r) =
\bigvee\nolimits_{\leq_\bool} \{ [\mwrc{\mathsf{B}_i}{I',J'}{s[\bar{\tt X}/\bar{x}, {\tt R}/r']}]_2
          \mid {r \leq r' }, s \in S_{\prog{P}'} \}$
\end{itemize}
and for any predicate ${\tt q}$ in $\Prog$ it is $I({\tt q})=I'({\tt q})$ and $J({\tt q})=J'({\tt q})$.
In addition, since $(I',J')$ is a stable fixpoint of $A_{\prog{P}'}$, it follows by Lemma~\ref{fixpoints}
that $(I,J)$ is a fixpoint of $A_\Prog$ and therefore $A_{\Prog}(\cdot,J)_1$
and $A_{\Prog}(I,\cdot)_2$
are well-defined monotone operators in the corresponding intervals.
It suffices to show that $I$ and $J$ are the least fixpoints,
\ie $I = \lfp \ATP(\cdot,J)_1$ and $J = \lfp \ATP(I,\cdot)_2$.


Let $I^* = \lfp A_{\Prog}(\cdot,J)_1$ and for the sake of contradiction assume it is $I^* < I$. 
Let ($I^*_\theta,J_\theta) = \tau(\theta(\tau^{-1}(I^*,J)))$. Specifically, for $I^*_\theta$ this means that
\begin{enumerate}
  \item for any common predicate ${\tt q}$ between $\Prog$ and $\prog{P}'$ it is 
$
{I^*_\theta}({\tt q}) = I^*({\tt q})
$ 
  \item for ${\tt test}_{\mathsf{B}_i}$ it is
${I^*_\theta}({\tt test}_{\mathsf{B}_i})({\bar{x}},r) =
  \bigvee\nolimits_{\leq_\bool} \{ [\mwrc{\mathsf{B}_i}{I^*,J}{s(\bar{\tt X}/\bar{x}, {\tt R}/r')} ]_1 \mid {r \leq r'}, s \in S_{\prog{P}'} \}
$.
\end{enumerate}
First, we show that $I^*_\theta < I'$. Indeed, 
for any predicate ${\tt q}$ other than ${\tt test_{\mathsf{B}_i}}$ we have $I^*_\theta({\tt q})=I^*({\tt q}) \leq I({\tt q})=I'({\tt q})$ and by 
our assumption
for at least one predicate ${\tt q'}$ the inequality is strict and  $I^*_\theta({\tt q'})=I^*({\tt q'}) < I({\tt q'})=I'({\tt q'})$. Furthermore,
since $I^* \leq I$ and because in the body $\mathsf{B}_i$ contains only predicates that appear in $P$ we have 
\begin{align*}
{I^*_\theta}({\tt test}_{\mathsf{B}_i})({\bar{x}},r) 
& = \bigvee\nolimits_{\leq_\bool} \{ [\mwrc{\mathsf{B}_i}{I^*,J}{s(\bar{\tt X}/\bar{x}, {\tt R}/r')} ]_1 \mid {r \leq r'}, s \in S_{\prog{P}'} \} \\
& \leq \bigvee\nolimits_{\leq_\bool} \{ [\mwrc{\mathsf{B}_i}{I,J}{s(\bar{\tt X}/\bar{x}, {\tt R}/r')} ]_1 \mid {r \leq r'}, s \in S_{\prog{P}'} \}\\
& \leq \bigvee\nolimits_{\leq_\bool} \{ [\mwrc{\mathsf{B}_i}{I',J'}{s(\bar{\tt X}/\bar{x}, {\tt R}/r')} ]_1 \mid {r \leq r'}, s \in S_{\prog{P}'} \} \\
& \leq I'({\tt test}_{\mathsf{B}_i})({\bar{x}},r) 
\end{align*}

We will show that $I^*_\theta$
is a fixpoint of $A_{\prog{P}'}(\cdot,J')_1$ which is a contradiction
since $I'$ is the least fixpoint. We proceed with case analysis on
the predicate ${\tt q}$. If ${\tt q}$ is not ${\tt p}$ or ${\tt test}_{\mathsf{B}_i}$ then
$A_{\prog{P}'}(I^*_\theta, J')_1({\tt q}) = A_{\prog{P}}(I^*,J)_1({\tt q})$
since the rules of ${\tt q}$ are the same in both $\Prog$ and $\prog{P}'$
and the predicates appearing in the bodies of these rules have 
the same meaning in $I^*$ and $I^*_\theta$ and the same meaning in $J'$ and $J$.
Since $I^*$ is a fixpoint of $A_{\prog{P}}$, it follows that
$A_{\prog{P}}(I^*, J)_1({\tt q}) = I^*({\tt q})=I^*_\theta({\tt q})$. Therefore,
$A_{\prog{P}'}(I^*_\theta, J')_1({\tt q}) = I^*_\theta({\tt q})$.
If ${\tt q}$ is ${\tt p}$ then,
\[
A_{\prog{P}'}(I^*_\theta, J')_1({\tt p})(\bar{x}) =
\bigvee\nolimits_{\leq_\bool}\{ [\mwrc{\mathsf{B}_j}{I^*_\theta,J'}{s[\bar{\tt X}/\bar{x}]}]_1 \mid j \neq i, s \in S_{\prog{P}'} \}
\vee I^*_\theta({\tt test}_{\mathsf{B}_i})(\bar{x}, \bot_{\rho_{\tt R}})
\]
By the definition of $I^*_\theta$ and since in each $\mathsf{B}_j$ 
the predicates appearing have the same meaning in $I^*$ and $I^*_\theta$ and the same meaning in $J'$ and $J$,
$I^*_\theta({\tt test}_{\mathsf{B}_i})(\bar{x}, \bot_{\rho_{\tt R}}) = \bigvee\nolimits_{\leq_\bool}\{ [\mwrc{\mathsf{B}_i}{I^*,J}{s[\bar{\tt X}/\bar{x}]}]_1
\mid s \in S_{\prog{P}'} \}= \bigvee\nolimits_{\leq_\bool}\{ [\mwrc{\mathsf{B}_i}{I^*_\theta,J'}{s[\bar{\tt X}/\bar{x}]}]_1
\mid s \in S_{\prog{P}'} \}$.
Therefore, we can rewrite the above as:
\[
A_{\prog{P}'}(I^*_\theta, J')_1({\tt p})(\bar{x}) =
\bigvee\nolimits_{\leq_\bool}\{ [\mwrc{\mathsf{B}_j}{I^*_\theta,J'}{s[\bar{\tt X}/\bar{x}]}]_1 \mid s \in S_{\prog{P}'} \}
\]
Since $\mathsf{B}_j$ does not contain ${\tt test}_{\mathsf{B}_i}$ and $S_\Prog = S_{\prog{P}'}$
it follows that
$[\mwrc{\mathsf{B}_j}{I^*_\theta,J'}{s[\bar{\tt X}/\bar{x}]}]_1 = [\mwrc{\mathsf{B}_j}{I^*,J}{s[\bar{\tt X}/\bar{x}]}]_1$ and so
$A_{\prog{P}'}(I^*_\theta, J')_1({\tt p})(\bar{x}) =
A_{\prog{P}}(I^*, J)_1({\tt p})(\bar{x})$ and since $I^*$ is a fixpoint
of $A_{\prog{P}}$, $A_{\prog{P}'}(I^*_\theta, J')_1({\tt p})(\bar{x})=
 A_{\prog{P}}(I^*, J)_1({\tt p})(\bar{x}) = I^*({\tt p})(\bar{x}) = I^*_\theta({\tt p})(\bar{x})$.
Finally, for ${\tt q}$ being ${\tt test}_{\mathsf{B}_i}$ we have
\begin{align*}
A_{\prog{P}'}(I^*_\theta,J')_1({\tt test}_{\mathsf{B}_i})(\bar{x},r)
& = \bigvee\nolimits_{\leq_\bool} \{ I^*_\theta({\tt test}_{\mathsf{B}_i})(\bar{x},r') \mid r\leq r', {|r'|\leq |r|+1} \} \\ 
& \ \ \ \ \vee 
  \bigvee\nolimits_{\leq_\bool}\{[\mwrc{\mathsf{B}_i}{I^*_\theta,J'}{s[\bar{\tt X}/\bar{x}, {\tt R}/r]}]_1 \mid  s \in S_{\prog{P}'} \}\\
& = \bigvee\nolimits_{\leq_\bool} \{ [\mwrc{\mathsf{B}_i}{I^*,J}{s(\bar{\tt X}/\bar{x}, {\tt R}/r')} ]_1 \mid {r \leq r'}, s \in S_{\prog{P}'} \}\\
& = I^*_\theta({\tt test}_{\mathsf{B}_i})(\bar{x},r')
\end{align*}
where we used the definition of $A_{\prog{P}'}$ for ${\tt test}_{\mathsf{B}_i}$ and replaced $I^*_\theta({\tt test}_{\mathsf{B}_i})$ 
in the second derivation. 
This establishes that $I^*_\theta$ is a fixpoint of $A_{\prog{P}'}(\cdot,J')_1$, which is a contradiction since $I^*_\theta < I'$
and $I'$ is its least fixpoint.

The case for $J$ is analogous to the case for $I$ and omitted.
\end{proof}

\begin{lemma}\label{same-three-valued-stable}
Let $S_3$ be the set of \emph{three-valued} stable models of $\Prog$ and let $S_3'$ be the set of \emph{three-valued} stable models of $\prog{P}'$. 
Then $\theta$ defines a bijection between $S_3$ and $S_3'$ with inverse $\theta^{-1}$. Let 
$S_2$ be the set of stable models of $\Prog$ and let $S_2'$ be the set of stable models of $\prog{P}'$.
Then $\theta$ also defines a bijection between $S_2$ and $S_2'$ with inverse $\theta^{-1}$.
\end{lemma}
\begin{proof}
Consequence of Lemma~\ref{rightdirection}, Lemma~\ref{leftdirection} and Lemma~\ref{fixpoints}.
\end{proof}

\begin{lemma}\label{same-wfs}
Let $\mathcal{M}_w$ be the well-founded model of $\Prog$ and let $\mathcal{M}_w'$ be the well-founded model of $\prog{P}'$. Then 
$\mathcal{M}_w = \theta^{-1}(\mathcal{M}_w')$ and $\mathcal{M}'_w = \theta(\mathcal{M}_w)$.
\end{lemma}
\begin{proof}
By Lemma~\ref{same-three-valued-stable}, $\theta,\theta^{-1}$ define a bijection between the sets of three-valued stable models and by
Lemma~\ref{theta-order-preserving}, they preserve the $\preceq$-order. Since the well-founded model is the $\preceq$-least 
\emph{three-valued} stable model the claim follows immediately.
\end{proof}


\begin{retheorem}{transformation-correctness-query}
Let $\Prog$ be a Higher-Order Datalog$^\neg$ program that defines a query
$\mathcal{Q}$ under
the well-founded semantics (resp. stable model semantics with cautious reasoning, stable model semantics with brave reasoning). 
Let $\prog{P}'$ be the program that results by applying the aforementioned
transformation to some rule of $\Prog$. Then, $\prog{P}'$ defines the same query $\mathcal{Q}$
under the well-founded semantics (resp. stable model semantics with cautious reasoning, stable model semantics with brave reasoning).
\end{retheorem}

\begin{proof}
Since the described transformation acts only on the query program it is clear that
for any input database $D$ the transformation of $\Prog \cup D$ results in $\prog{P}' \cup D$ and thus the bijection 
described in Lemma~\ref{same-three-valued-stable} holds for the stable models of the enhanced programs with the facts supplied from the database.
Furthermore, $\prog{P'} \cup D$  has a non-empty set of stable models if and only if 
$\Prog \cup D$ has a non-empty set of stable models.

For any predicate ${\tt q}$ of $\Prog$ we show that it is $\mathcal{M}({\tt q}) (\bar{a})=\mtrue$ for every stable model $\mathcal{M}$ of $\Prog \cup D$ 
if and only if  $\mathcal{M}'({\tt q}) (\bar{a})=\mtrue$ for every 
stable model $\mathcal{M}'$  of $\prog{P}' \cup D$, which establishes equivalence under stable model semantics and cautious reasoning.
Assume for the sake of contradiction that there exists a stable model $\mathcal{M}'$ of $\prog{P}' \cup D$ such that ${\mathcal{M}'}({\tt q}) (\bar{a})=\mfalse$.
By Lemma~\ref{same-three-valued-stable} then $\theta^{-1}(\mathcal{M}')$ is a stable model for $\Prog \cup D$ and 
${\theta^{-1}(\mathcal{M}')}({\tt q}) (\bar{a})=\mfalse$ (contradiction).
Likewise, assume that there exists a stable model $\mathcal{M}$ of $\Prog \cup D$ such that ${\mathcal{M}}({\tt q}) (\bar{a})=\mfalse$. 
By Lemma~\ref{same-three-valued-stable} then $\theta(\mathcal{M})$ is a stable model for $\prog{P}' \cup D$ and 
${\theta(\mathcal{M})}({\tt q}) (\bar{a})=\mfalse$ (contradiction).

Similarly, it is easy to show that there exists a stable model $\mathcal{M}$ of $\Prog \cup D$ such that $\mathcal{M}({\tt q}) (\bar{a})=\mtrue$ if and only
if there exists a stable model $\mathcal{M}'$ of $\prog{P}'$ such that ${\mathcal{M}'}({\tt q}) (\bar{a})=\mtrue$, thus establishing equivalence under
stable model semantics and brave reasoning.

Finally, we can easily show that $\mathcal{M}_w({\tt q}) (\bar{a})=\mtrue$ in the well-founded model $\mathcal{M}_w$ of $\Prog \cup D$ if and only if
$\mathcal{M}_w({\tt q}) (\bar{a})=\mtrue$ in the well-founded model $\mathcal{M}_w'$ of $\prog{P}' \cup D$. This claim follows from Lemma~\ref{same-wfs}.
\end{proof}
\fi

\end{document}